\documentclass[lettersize,journal]{IEEEtran}
\usepackage{amsmath,amsfonts}
\usepackage{algorithmic}
\usepackage{algorithm}
\usepackage{array}
\usepackage[caption=false,font=normalsize,labelfont=sf,textfont=sf]{subfig}
\usepackage{textcomp}
\usepackage{stfloats}
\usepackage{url}
\usepackage{verbatim}
\usepackage{graphicx}
\usepackage{cite}
\newcounter{function}
\newcounter{algorithm saved}

\makeatletter
\newenvironment{function}[1][htb]{%
    \renewcommand{\ALG@name}{Func.}
    \setcounter{algorithm saved}{\value{algorithm}} 
    \setcounter{algorithm}{\value{function}}
    \begin{algorithm}[#1]%
    }{\end{algorithm}
    \setcounter{function}{\value{algorithm}}
    \setcounter{algorithm}{\value{algorithm saved}}
}
\makeatother

\usepackage{amsthm}
\usepackage{enumerate}
\usepackage{multirow}
\usepackage{float}
\usepackage{xcolor}
\usepackage{newfloat}
\usepackage{adjustbox}

\usepackage{booktabs}
\usepackage{array}

\DeclareFloatingEnvironment[
    fileext=los,
    listname=List of Protocol,
    name=Protocol,
    placement=H,
]{protocol}

\newtheorem{theorem}{Theorem}[section]

\usepackage{url}
\usepackage{bm}

\begin{document}

\title{Efficient Dropout-resilient Aggregation for Privacy-preserving Machine Learning}

\author{Ziyao Liu\textsuperscript{\textsection}, Jiale Guo\textsuperscript{\textsection}, Kwok-Yan Lam~\IEEEmembership{Senior Member,~IEEE,}, and~Jun~Zhao~\IEEEmembership{Member,~IEEE}
\IEEEcompsocitemizethanks{\IEEEcompsocthanksitem Ziyao Liu, Jiale Guo, Kwok-Yan Lam, and Jun Zhao are with the School of Computer Science and Engineering, Nanyang Technological University, Singapore, 50 Nanyang Ave, 639798.\protect ~E-mail: \{ziyao002, jiale001\}@e.ntu.edu.sg, \{kwokyan.lam, junzhao\}@ntu.edu.sg. Corresponding author: Jun Zhao}
\thanks{This research is supported in part by the National Research Foundation, Singapore under its Strategic Capability Research Centres Funding Initiative. Any opinions, findings and conclusions or recommendations expressed in this material are those of the author(s) and do not reflect the views of National Research Foundation, Singapore. This research is also supported by in part by Nanyang Technological University (NTU)
Startup Grant, and Singapore Ministry of Education Academic Research Fund under Grant Tier 2 MOE2019-T2-1-176.}

\thanks{Manuscript received September 02, 2021; revised January 9, 2022 and February 18, 2022; accepted March 14, 2022. Date of publication TBD, 2022. }}

\markboth{IEEE TRANSACTIONS ON INFORMATION FORENSICS AND SECURITY,~Vol.~TBD, No.~TBD, TBD~2022}%
{Ziyao \MakeLowercase{\textit{et al.}}: Efficient Dropout-resilient Aggregation for Privacy-preserving Machine Learning}

\IEEEpubid{\begin{minipage}[t]{\textwidth}\ \\
 \centering Copyright \copyright 2022 IEEE. Personal use of this material is permitted. Permission from IEEE must be obtained for all other uses, in any current or future media, including reprinting/republishing this material for advertising or promotional purposes, creating new collective works, for resale or redistribution to servers or lists, or reuse of any copyrighted component of this work in other works.\end{minipage}}

\maketitle
\begingroup\renewcommand\thefootnote{\textsection}
\footnotetext{Equal contribution}
\endgroup

\begin{abstract}
Machine learning (ML) has been widely recognized as an enabler of the global trend of digital transformation. With the increasing adoption of data-hungry machine learning algorithms, personal data privacy has emerged as one of the key concerns that could hinder the success of digital transformation. As such, Privacy-Preserving Machine Learning (PPML) has received much attention of the machine learning community, from academic researchers to industry practitioners to government regulators. However, organizations are faced with the dilemma that, on the one hand, they are encouraged to share data to enhance ML performance, but on the other hand, they could potentially be breaching the relevant data privacy regulations. Practical PPML typically allows multiple participants to individually train their ML models, which are then aggregated to construct a global model in a privacy-preserving manner, e.g., based on multi-party computation or homomorphic encryption. Nevertheless, in most important applications of large-scale PPML, e.g., by aggregating clients' gradients to update a global model for federated learning, such as consumer behavior modeling of mobile application services, some participants are inevitably resource-constrained mobile devices, which may drop out of the PPML system due to their mobility nature \cite{yang2019federated}. Therefore, the resilience of privacy-preserving aggregation has become an important problem to be tackled because of its real-world application potential and impacts. In this paper, we propose a scalable privacy-preserving aggregation scheme that can tolerate dropout by participants at any time, and is secure against both semi-honest and active malicious adversaries by setting proper system parameters. By replacing communication-intensive building blocks with a seed homomorphic pseudo-random generator, and relying on the additive homomorphic property of Shamir secret sharing scheme, our scheme outperforms state-of-the-art schemes by up to 6.37$\times$ in runtime and provides a stronger dropout-resilience. The simplicity of our scheme makes it attractive both for implementation and for further improvements.
\end{abstract}


\section{Introduction}
\IEEEPARstart{W}{ith} the widespread adoption of data-hungry machine learning algorithms and increasing concerns for personal data privacy protection, Privacy-Preserving Machine Learning (PPML) has emerged as an important area that received much attention of the machine learning community, from academic researchers to industry practitioners to government regulators \cite{yang2021privacy,li2021federated,nguyen2021flguard,hesamifard2018privacy}. PPML allows multiple participants, e.g., data owners, to jointly solve a machine learning problem while preserving their data privacy. Traditional PPML typically enables data owners to individually perform their ML to train a model using their local data, i.e., compute the gradients, which are then aggregated to construct a global model. However, as pointed out in \cite{zhu2019deep}, local data of an individual participant could be revealed through a small portion of gradients of its local model. This severe leakage from gradients even allows the attacker to recover images with pixel-wise accuracy and texts with token-wise matching. In this connection, cryptographic mechanisms such as secure multi-party computation (MPC) \cite{goldreich1998secure,byali2020flash} and homomorphic encryption (HE) \cite{gentry2009fully,brutzkus2019low} have been proposed to enhance PPML by aggregating the local models in a privacy-preserving manner. 

\IEEEpubidadjcol
In general, one can enhance PPML by constructing a secure\footnote{We use the terms secure and privacy-preserving interchangeably.} aggregation scheme that protects the local models by privacy-preserving technology (PPT) such as differential privacy (DP) \cite{dwork2014algorithmic,yang2020local} and the aforementioned cryptographic mechanisms (i.e., MPC and HE). For example, DP can be applied to clients' gradients before they are uploaded to the aggregation server \cite{zhao2020local}. In this way, the privacy of gradient information is protected while the server can still aggregate the perturbed gradients to obtain an approximate result according to the properties of DP. Nevertheless, this method suffers from the trade-off between privacy protection and data usability, hence model performance. Whereas, for HE-based aggregation schemes such as \cite{chan2012privacy,aono2017privacy,sinem21poseidon}, participants perform compute-intensive algorithms to encrypt their gradients and send to the server. Then the server aggregates the protected gradients by performing arithmetic operations on the ciphertext and sends them back to the participants for decryption or continuing the training over ciphertext. For example, to achieve a similar training accuracy for CIFAR10 image classification, it requires 175 hours in POSEIDON \cite{sinem21poseidon}, but only about 1.67 hours using our scheme with similar neural network architecture and server configuration.

With the rapid growth of the digital economy involving billions of users in Cyberspace, to ensure business sustainability in a highly competitive environment, most application service providers have adopted large-scale PPML to model user behavior and preferences in order to provide improved user experiences. It is worth noting that, a large number of users are likely to be on mobile devices, so PPML needs to be highly resilient and cope with the dynamic connectivity of mobile devices. For example, Gboard on Android (the Google Keyboard) \cite{yang2018applied} has deployed a PPML on mobile phones to improve the accuracy of suggested queries for the user's current typing context. Furthermore, researchers from Apple have used PPML deployed on iPhones to enhance the performance of speaker verification \cite{granqvist2020improving}. In this case, large-scale PPML inevitably involves resource-constrained mobile devices which may drop out of the system due to their mobility nature.

In this connection, dropout-resilient privacy-preserving aggregation, e.g., by aggregating clients' gradients to update a global machine learning model in a privacy-preserving manner for Federated Learning \cite{kairouz2019advances}, has attracted tremendous attention of the research community because of its real-world application potential and impacts. For example, \cite{chan2012privacy} proposed a scheme, based on threshold HE cryptosystem, to securely share the private key and to deal with dropout clients. However, this scheme requires all the clients to contribute to several expensive building blocks such as distributed key generation and decryption \cite{hazay2019efficient}, which is not practical for large-scale PPML. Compared with \cite{chan2012privacy}, the scheme in \cite{aono2017privacy} is more efficient by having the secret key held by every participant, though it sacrifices the privacy requirements of PPML that the privacy of gradient must be kept by the corresponding participant. On the other hand, a pure MPC-based aggregation scheme is not suitable for a large-scale PPML due to the huge communication overheads when evaluating complex functions such as a deep neural network (DNN), and the problem is further exacerbated by the fact that a lot of the PPML clients are resource-constrained mobile devices. For example, training a simple CNN one epoch requires about 7 hours in WAN setting \cite{liu2020mpc}. Server-aided MPC-based schemes such as \cite{mohassel2018aby3,wagh2020falcon} achieve good efficiency relying on a set of non-colluding servers. However, for PPML systems that are deployed by a single agent such as a commercial company or government, while the trust distribution is limited to the number of aided servers, it is not easy to guarantee that the involved aided servers are non-colluding from a game-theoretical perspective.

Recent advancements in dropout-resilient secure aggregation significantly improved the protocol efficiency by leveraging on pair-wise additive masking and Shamir secret sharing as proposed in the pioneering work SecAgg \cite{bonawitz2017practical}.  
Assume there is a set of client $\mathcal{U}$, and let each client $u_i \in \mathcal{U}$ holds a vector $\boldsymbol{x}_i$, the goal is to calculate $\sum \boldsymbol{x}_i$ while preserving the privacy of $\boldsymbol{x}_i$. In \cite{bonawitz2017practical}, a pair-wise additive mask is added to $\boldsymbol{x}_i$ (assume a total order on clients) and the client $u_i$ uploads $\boldsymbol{y}_{i}$ to the central server rather than $\boldsymbol{x}_i$:
$$\boldsymbol{y}_{i}=\boldsymbol{x}_{i}+\sum_{i<j} \text{PRG}(s_{i, j})-\sum_{i>j} \text{PRG}(s_{j, i})$$
where PRG denotes a pseudorandom generator that is able to generate a sequence of random numbers using the seed $s_{i, j}$. It is obvious from the equation that, when aggregating all the $\boldsymbol{x}_i$, the masks will be canceled such that

\begin{footnotesize}
\begin{equation*}
\begin{aligned}
\sum_{u_i\in \mathcal{U}}\boldsymbol{y}_{i}=\sum_{u_i\in \mathcal{U}}\left(\boldsymbol{x}_{i}+\sum_{i<j} \text{PRG}(s_{i, j})-\sum_{i>j} \text{PRG}(s_{j, i})\right)=\sum_{u_i\in \mathcal{U}} \boldsymbol{x}_{i}    
\end{aligned}
\end{equation*}
\end{footnotesize}

Furthermore, the seeds are shared among the clients using the standard $t$-out-of-$n$ Shamir secret sharing scheme (see Section \ref{sec:Shamir-secret-sharing}) to handle the dropout clients. To generate the pair-wise masks, each pair of client $(u_i,u_j)$ are involved with Diffie-Hellman (DH) based key exchange protocol \cite{diffie2019new} to make an agreement on the seed $s_{i, j}$. Note that, with an $n$-clients PPML system, running a pair-wise DH protocol is not inexpensive as it has $\mathcal{O}(n^2)$ communication-round complexity.  The follow-up works such as \cite{mandal2019privfl} and \cite{guo2021privacy} also show its inefficiency. Besides, the runtime of the whole protocol increases rapidly as the fraction of dropped participants (dropout rate) increases, as shown in Fig. \ref{fig:my_label}. In this case, active malicious dropping which slows down the system could even be exploited as a kind of attack.
\begin{figure}[htbp]
    \centering
    \includegraphics[width=0.7\linewidth]{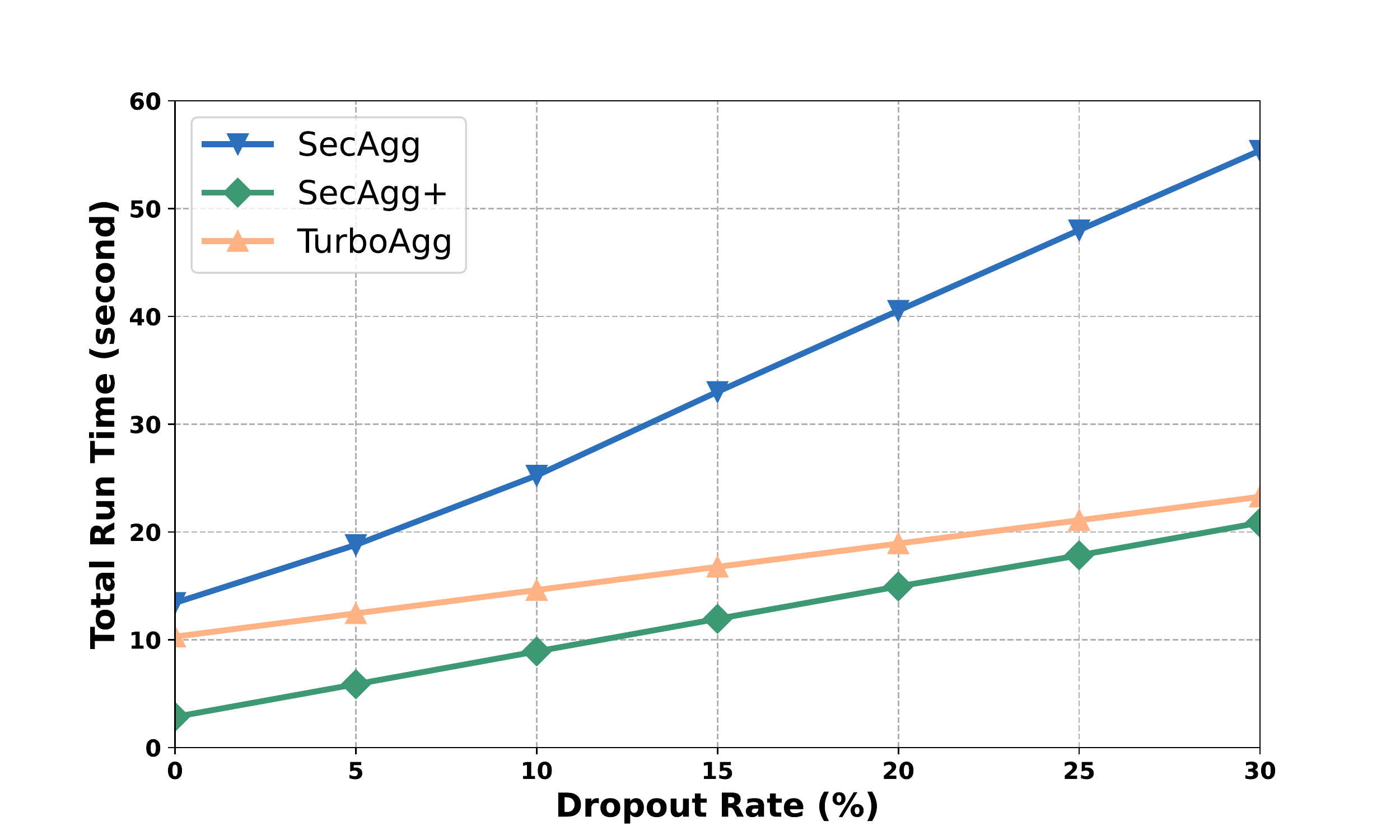}
    \caption{The runtime of the state-of-the-art aggregation protocols including TurboAgg \cite{so2021turbo}, SecAgg \cite{bonawitz2017practical} and SecAgg+ \cite{bell2020secure} with the increase of the dropout rate.}
    \label{fig:my_label}
\end{figure}

To improve the efficiency of the SecAgg scheme, there are three possible directions proposed by previous works: 
\begin{enumerate}[-]
    \item Communicate across only a subset of clients: For example, the variant TurboAgg \cite{so2021turbo} divides clients into $l$ groups with $n_l$ clients of each group and follows a multi-group circular structure. Both CCESA scheme \cite{choi2020communication} and SecAgg+ \cite{bell2020secure} replace the star topology of the communication network in SecAgg with random subgroups of clients, i.e., a $k$-regular graph such as Erdos-Renyi graph and Harray graph, and thus reduce the communication rounds between clients.
    \item Optimize the computational complexity: For example, FastSecAgg \cite{kadhe2020fastsecagg} substitutes standard Shamir secret sharing with a more efficient FFT-based multi-secret sharing scheme, thus reducing the computational cost of both the server and clients.
    \item Replace pair-wise DH protocol with other schemes: For example, NIKE \cite{mandal2018nike} adopts a non-interactive key exchange protocol, thus improving the efficiency of the generation of cryptographic materials for executing security protocols.
    \item Compress the gradient vector: For example, SAFER \cite{beguier2020efficient} proposes an MPC-friendly coding method to compress the gradient vector, and thus to reduce the communication cost between the clients and server.
\end{enumerate}

We note that TurboAgg and FastSecAgg sacrifice some security for efficiency purposes. NIKE involves two non-colluding servers to construct a $2$-out-of-$3$ Shamir secret sharing scheme, which is not suitable for PPML applications that rely on a non-collusion assumption. SAFER does not support large-scale systems and has no dropout-resilience. In addition, a more generic scheme SAFELearn \cite{fereidooni2021safelearn} can be instantiated by FHE or MPC to construct efficient private FL systems. However, as demonstrated in Section IV.B of SAFELearn \cite{fereidooni2021safelearn}, its FHE instantiation is computationally expensive, and thus is not friendly for PPML systems that involve resource-constrained mobile devices. The more practical MPC-based SAFELearn still relies on non-colluding servers, which limits its application scenarios to that with a non-collusion assumption. Therefore, considering the computation-friendly schemes that involve only one single server while providing both dropout-resilience and security against active malicious adversaries, we choose SecAgg as the baseline and SecAgg+ as state-of-the-art for a fair comparison. Besides, We summarize the performance of the aggregation schemes that do not rely on non-collusion assumptions, i.e., do not involve non-colluding servers, in Table \ref{tab:comp_table}.

\begin{table*}
\scriptsize
\setlength\tabcolsep{3pt}
\caption{Comparison of the computation complexity, communication complexity, dropout resilience and privacy guarantee between SecAgg \cite{bonawitz2017practical}, TurboAgg \cite{so2021turbo}, CCESA \cite{choi2020communication}, SecAgg+ \cite{bell2020secure} and FastSecAgg \cite{kadhe2020fastsecagg}. Here $n$ is the total number of clients, and $m$ is the length of each client's vector. In TurboAgg, $n_l$ is set to be $\log n$ as the group size for the best trade-off. $k$ is set to be $\mathcal{O}(\log{n})$ as the degree of each client in $k$-regular graph in SecAgg and $\mathcal{O}(\sqrt{n/\log{n}})$ in CCESA, $\delta$ is the dropout rate. In FastSecAgg, $d$ can be set up to $\frac{n}{2}$. Note that if the protocol listed in the table provides security against malicious adversaries, its security is also guaranteed against semi-honest adversaries. Since the maximum number of dropout clients is not the same for different privacy guarantees, we only give the comparison in a semi-honest setting. Note that some of complexity analysis are extracted from SAFELearn paper \cite{fereidooni2021safelearn}, and we assume $m$ is greater than $n$ for many real-world scenarios.}
\centering
\begin{tabular}{c|c|cccccc}
\toprule[1pt]
\multicolumn{2}{c|}{Protocol}                                                 & SecAgg \cite{bonawitz2017practical}& TurboAgg \cite{so2021turbo} & CCESA \cite{choi2020communication} & SecAgg+ \cite{bell2020secure} & FastSecAgg \cite{kadhe2020fastsecagg} & Our's \\ 
\midrule
\multirow{2}{*}{\begin{tabular}[c]{@{}c@{}}Computation\\ complexity\end{tabular}}&Server& $\mathcal{O}(mn^2)$   &$\mathcal{O}(m\log n \log^2 \log n)$ & $\mathcal{O}(mn\log n)$ &$\mathcal{O}(mn\log n +n\log^2 n)$&$\mathcal{O}(m\log n)$&$\bm{\mathcal{O}(n)}$\\ 
\specialrule{0em}{1pt}{0pt}
&Client& $\mathcal{O}(n^2+mn)$ &$\mathcal{O}(m\log n \log^2 \log n)$& $\mathcal{O}(n\sqrt{n\log n}+mn)$  &$\mathcal{O}(m\log n+\log^2 n)$&$\mathcal{O}(m\log n)$&$\bm{\mathcal{O}(n^2+m)}$\\ 
\midrule
\specialrule{0em}{1pt}{0pt}
\multirow{2}{*}{\begin{tabular}[c]{@{}c@{}}Communication\\ complexity\end{tabular}}&Server&$\mathcal{O}(n^2+mn)$&$\mathcal{O}(mn\log n)$&$\bm{\mathcal{O}(n\log n+m\sqrt{n\log n})}$&$\mathcal{O}(mn+n\log n)$&$\mathcal{O}(n^2+mn)$&$\mathcal{O}(n^2+mn)$\\ 
\specialrule{0em}{1pt}{0pt}
&Client&$\mathcal{O}(m+n)$&$\mathcal{O}(m\log n)$&$\mathcal{O}(\sqrt{n\log n+m})$&$\bm{\mathcal{O}(m+\log n)}$&$\mathcal{O}(m+n)$&$\mathcal{O}(m+n)$\\ 
\midrule
\specialrule{0em}{1pt}{0pt}
\multirow{2}{*}{\begin{tabular}[c]{@{}c@{}}Dropout\\ resilience\end{tabular}} &Scheme&$(t, n)$&$(\frac{n_l}{2}, n_l)$&$(t, k)$&$(t, k)$&$(n-d, n)$&$(t, n)$\\ 
\specialrule{0em}{1pt}{0pt}
&Max drop&$n-1$&$\frac{n}{2}-1$&$\delta n$&$\delta n$&$\frac{n}{2}-1$&$n-1$\\ 
\midrule
\specialrule{0em}{1pt}{0pt}
\multicolumn{2}{c|}{Communication rounds}&4&$n/\log n$&3&3&3&3\\ 
\midrule
\specialrule{0em}{1pt}{0pt}
\multicolumn{2}{c|}{Privacy}&malicious&semi-honest&semi-honest&malicious&semi-honest&malicious\\ 
\bottomrule[1pt]
\end{tabular}
\label{tab:comp_table}
\end{table*}

In this paper, leveraging on homomorphic pseudorandom generator (HPRG) and Shamir secret sharing scheme, we propose an efficient and dropout-resilient aggregation scheme for PPML. In general, the improved efficiency and dropout-resilience mainly come from the application of HPRG and the sophisticated integration with Shamir secret sharing and the building blocks of SecAgg-based schemes. Specifically, we replace the Diffie-Hellman key exchange protocol for seed agreement in existing schemes with an HPRG-based scheme. In this case, no interaction is needed for the clients to construct the seed used for generating masks, which significantly reduces the communication overheads compared to previous works. Furthermore, since the connected clients do not need to send additional Shamir shares to the server to reconstruct the dropped users' secrets, compared to SecAgg-based schemes, less interaction over clients leads to lower computation and communication overheads. Meanwhile, the server needs only one computation for seed reconstruction rather than the number of dropout clients in SecAgg-based schemes. Thus, our proposed scheme has a stronger dropout-resilience that the runtime of the whole protocol decreases with the increase of the dropout rate.

Compared with existing works, our contributions are as follows:

\begin{enumerate}[-]
    \item The proposed scheme is more efficient, which significantly reduces the communication and computation overheads.
    \item The proposed scheme has a stronger dropout-resilience that the runtime of the whole protocol decreases with the increase of dropout rate. At the same time, the $t$-out-of-$n$ Shamir secret sharing scheme is kept, and no extra trusted third party is needed.
    \item The simplicity of the proposed scheme makes it attractive both for implementation and for further improvements.
\end{enumerate}

Moreover, it can be proven that our scheme is secure against both semi-honest and active malicious adversaries by setting proper system parameters even if a set of clients drops out of the protocol at any time.

\textbf{Organisation of the paper:} The rest of the paper is organized as follows. In Section \ref{sec:preliminaries}, we review the underlying cryptographic primitives and define the notations used by our proposed scheme. Then we proceed to our proposed protocol in Section \ref{sec:high-level overview}, followed by the security analysis in Section \ref{sec:sec-analysis}, performance analysis and discussions in Section \ref{sec:performance-analysis}. Finally, we give the conclusions in Section \ref{sec:conclusion}.

\section{Underlying Cryptographic Primitives} 
\label{sec:preliminaries}
This section briefly describes the preliminaries of Shamir secret sharing scheme, homomorphic pseudorandom generator, and signature schemes.

\subsection{Shamir Secret Sharing}
\label{sec:Shamir-secret-sharing}
Shamir secret sharing scheme \cite{shamir1979share} divides a secret $S$ into $n$ pieces of data called shares such that (i) the secret $S$ can be efficiently reconstructed by any combination of $t$ data pieces, and (ii) the secret $S$ cannot be reconstructed by any set of data pieces of which the number is less than $t$. Such scheme is called a $t$-out-of-$n$ or ($t$, $n$) threshold scheme.

In specific, for a standard ($t$, $n$) Shamir secret sharing scheme, the secret $S$ and shares $S_1, \dots, S_n$ are the elements in a finite field $\mathbb{Z}_P$ for some prime $P$ where $0<t\leq n<P$. We assume that there is one secret holder $u_s$ and $n$ participants $\{u_1, \dots, u_n\}$. The scheme works as follows:

\begin{enumerate}
    \item Preparation: The secret holder $u_s$ randomly chooses $t-1$ positive integers $a_1, \dots, a_{t-1}$ from $\mathbb{Z}_P$ and $a_0=S$ to define a polynomial of degree $t-1$, i.e., $f(x)=a_0+a_1x+a_2x^2+a_3x^3+\dots+a_{t-1}x^{t-1} \bmod{P}$.
    \item Secret sharing: The secret holder $u_s$ randomly chooses $n$ points $x_1, \dots, x_n$ to retrieve $\{x_i, f(x_i)\}$ for $i \in \{1,2,\dots,n\}$, and sends them to the corresponding participants $u_i$.
    \item Secret reconstructing: Given any $t$ of $\{x_i, f(x_i)\}$, the secret holder is able to calculate the coefficients $a_0, \dots, a_{t-1}$ of the polynomial $f(x)$ using Lagrange interpolation, and the constant term $a_0$ is the secret. A more efficient method to directly reconstruct the secret is to calculate $S=a_0=\sum_{j=0}^{t} f\left(x_{j}\right) \prod_{m=0, m \neq j}^{t} \frac{x_{m}}{x_{m}-x_{j}} \bmod{P}$.
\end{enumerate}

For simplicity, we denote $\mathcal{F}_{gen}$ as the share generation function to generate shares $\{x_i, f(x_i)\}$, and $\mathcal{F}_{rec}$ as the secret reconstruction function to reconstruct the secret $S$ for $(t,n)$ standard Shamir secret sharing.

Shamir secret sharing scheme is additive homomorphic \cite{goldreich1998secure}. For example, assume that the party $P_1$ has a secret $S_1$. To secretly share $S_1$ among $n$ parties $P_1, P_2, \dots, P_n$, the secret holder $P_1$ chooses a $t-1$ degree polynomial $f(x)=a_0+a_1x+\dots+a_{t-1}x^{t-1} \bmod{P}$ where $a_0=S_1$ and $a_1, \dots, a_{t-1}$ are randomly chosen from $\mathbb{Z}_P$, then computes each Shamir share of $S_1$ to be sent to $P_i$ as $\{x_i,f(x_i)\}$ for $x_i\in X=\{x_1, x_2, \dots, x_n\}$. Similarly, another secret holder $P_2$ who has the secret $S_2$ chooses a $t-1$ degree polynomial $g(x)=b_0+b_1x+\dots+b_{t-1}x^{t-1}$, where $b_0=S_2$ and $b_1, \dots, b_{t-1}$ are randomly chosen from $\mathbb{Z}_P$, then computes $n$ shares of $S_2$ as $\{x_i,g(x_i)\}$ for $x_i\in X$ that $\{x_i,g(x_i)\}$ is sent to $P_i$. In this case, each party $P_i$ can locally compute $\{x_i,f(x_i)+g(x_i)\}$, and then cooperates with other parties of which the number is greater than $t$ to reconstruct the secret $S_1+S_2$ using Lagrange interpolation over $\{x_i,f(x_i)+g(x_i)\}$ for $x_i \in X$. We note that addition over Shamir shares works only when $P_1$ and $P_2$ agree on the same set $X$, which is usually assigned to be $\{1,2, \dots, n\}$ in Shamir based applications \cite{bonawitz2017practical, bell2020secure, so2021turbo}. Besides, when context is clear, we abuse $f(x_i)$ to denote the share rather than $\{x_i, f(x_i)\}$.

\subsection{Homomorphic Pseudorandom Generator}
\label{sec:homo-prg}

A pseudorandom function (PRF) is an efficient algorithm that approximately maps two distinct sets $F: \mathcal{K}\times\mathcal{X}\rightarrow\mathcal{Y}$ such that a uniform $k \in \mathcal{K}$, a uniform function $f:\mathcal{X}\rightarrow\mathcal{Y}$, an oracle for $F(k, \cdot)$ is computationally indistinguishable from an oracle for $f(\cdot)$ \cite{goldreich2019construct}. Similar to the definition of PRF, a pseudorandom generator (PRG) is an efficient algorithm that is able to generate a sequence of approximate random numbers. More specifically, PRG is an efficient computable function $G:\mathcal{S}\rightarrow\mathcal{Y}$ such that for uniform $s \in \mathcal{S}$ and uniform $y \in \mathcal{Y}$, the distribution $\{G(s)\}$ is computationally indistinguishable from the distribution of $\{y\}$.

A pseudorandom function $F: \mathcal{K}\times\mathcal{X}\rightarrow\mathcal{Y}$ is said to be key homomorphic if for any $F(k_1, x)$ and $F(k_2, x)$, an efficient algorithm exists to compute $F(k_1 \oplus k_2, x)=F(k_1, x) \otimes F(k_2, x)$ where both $(\mathcal{K, \oplus})$ and $(\mathcal{Y, \otimes})$ are groups. In simple words, the PRF is homomorphic with respect to its key. Similarly, a PRG function $G: \mathcal{S}\rightarrow\mathcal{Y}$ is said to be seed homomorphic if for any $G(s_1, x)$ and $G(s_2, x)$, we have $G(s_1 \oplus s_2, x)=G(s_1, x) \otimes G(s_2, x)$. Such type of PRG is called homomorphic PRG (HPRG).

As described in \cite{banerjee2015key}, constructing an HPRG is quite straightforward in random oracle model. Let $(\mathbb{G},q)$ be a finite cyclic group of prime order $q$, and $H: \mathcal{X}\rightarrow \mathbb{G}$ be a hash function modeled as random oracle, and the function $F: \mathbb{Z}_q \times \mathcal{X}\ \rightarrow \mathbb{G}$ be $$F(k,x)=H(x)^k$$
We can observe that $F$ is key homomorphic:
$$F(k_1+k_2,x)=F(k_1,x)\cdot F(k_2,x)$$
It is proved in \cite{naor1999distributed} that if the Decision Diffie-Hellman (DDH) assumption holds in $\mathbb{G}$, such function $F$ is secure in random oracle model that if $k$ is uniform in $\mathcal{K}$, then $F(k,\cdot)$ is indistinguishable from a random sample in $\mathbb{G}$. Therefore, a simple version of HPRG can be constructed by using the seed $s$ for $k$ and indices for $x$. In specific, HPRG $G(s)$ can generate a sequence of elements, i.e.,  $F(s,1),F(s,2),F(s,3)...$. The length of such sequence is determined by the vector to be masked. For example, if the input vector $\boldsymbol{x}=[x_1,x_2,...,x_m]$ has $m$ entries, then $G(s)$ generates $F(s,1),F(s,2),...,F(s,m)$, and the masking operations are done element-wisely.

Although there are some other more robust HPRG in the standard model, such as \cite{boneh2013key} and \cite{banerjee2014new}, they rely on learning with errors assumptions, thus are not practical for their inefficient parameter size and runtime. In this work, we choose DDH based HPRG for its good trade-off between efficiency and security.

\subsection{Signature Scheme}
\label{sec:signature}
A signature scheme is used to prove the origin of a message. If a message is signed by Alice's secret key, then the message must have come from Alice. A signature scheme is a tuple of algorithm $(\mathcal{F}_{kg}, \mathcal{F}_{sig}, \mathcal{F}_{vrfy})$ such that
\begin{itemize}
    \item $\mathcal{F}_{kg}$ is a randomized algorithm that outputs a secret key $sk$ and a public key $pk$.
    \item $\mathcal{F}_{sig}$ is a randomized algorithm that receives the secret key $sk$ and a message $m$ and outputs a signature $\sigma$.
    \item $\mathcal{F}_{vrfy}$ is a deterministic algorithm that takes in the public key $pk$, the message $m$, and the signature $\sigma$, and returns one if $\sigma$ is a signature on $m$ and 0 otherwise.
\end{itemize}

The signature scheme achieves security against universal forgery under chosen message attack (UF-CMA). It means that someone without the secret key can not create a valid signature on a message he has not seen signed before. In other words, the probability of the adversary to construct a pair $(m^*, \sigma^*)$ without knowing the secret key $sk$ that $\sigma^*$ is a valid signature on $m^*$ and $m^*$ has never been seen before is negligible.

\section{Proposed Scheme}
\label{sec:high-level overview}

Similar to the SecAgg scheme \cite{bonawitz2017practical}, we divide participants into two classes: (i) a central server $\mathcal{S}$ that acts as a coordinator to aggregate inputs from $n$ clients $\mathcal{U}$, and (ii) each client $u_i \in \mathcal{U}$ holds a locally trained model $\boldsymbol{x}_i$. The goal of our scheme is that the server can compute the sum of clients' models as $z=\sum_{u_i \in \mathcal{U}}\boldsymbol{x}_i$ while keeping the privacy of $\boldsymbol{x}_i$ only to the client $u_i$. Besides, the scheme should be dropout-resilient as clients may drop out at any time.

\textbf{Threat model.} In many PPML applications such as FL, the participants can be individual users or competitive business entities and are required to comply with the data privacy regulations. Besides, some participants may cooperate in invading other's data privacy for their benefit. Due to these natures, we consider two threat models. In the first threat model, adversary participants are semi-honest that will not deviate from the protocol but try to infer the honest parties' information. In the second threat model, adversary participants can be active malicious and may collude and send fraudulent messages to others. We note that our proposed protocol is secure against semi-honest adversaries, and provides additional security against active malicious adversaries by adopting specific security protocols of which the number of adversaries has an upper bound (see Section \ref{sec:malicious-model} for the details). Our security definition is based on the Universal Composability (UC) framework, and we refer interested readers to \cite{canetti2001universally} for the details.

\subsection{Masking Models}
\label{sec:masking-models}

Since we aim to protect the clients' locally trained models from information leakage, the most intuitive method is to mask the model using a one-time pad before the models are submitted to the central server. Note that the mask is added to the model element-wisely to permute the model's distribution. Otherwise, since the range of the model's elements is known for a fixed ML model, adversaries can easily reconstruct the original model if all the model's elements are masked by the same random value. Let the client $u_i$'s model be $\boldsymbol{x}_i$, and its self-generated mask be $\boldsymbol{r}_i$, then client $u_i$'s masked model can be locally calculated by:

$$\boldsymbol{y}_i=\boldsymbol{x}_i+\boldsymbol{r}_i $$

After that, the client sends $\boldsymbol{y}_i$ to the server. Assume that the server already had the sum of $\boldsymbol{r}_i$ for all the clients as $\boldsymbol{R}=\sum_{u_i \in \mathcal{U}} \boldsymbol{r}_{i}$, the server can compute:

$$\boldsymbol{z}=\sum_{u_i \in \mathcal{U}} \boldsymbol{y}_{i}-\boldsymbol{R}=\sum_{u_i \in \mathcal{U}} \boldsymbol{y}_{i}-\sum_{u_i \in \mathcal{U}} \boldsymbol{r}_{i}=\sum_{u_i \in \mathcal{U}}\boldsymbol{x}_i $$

The naive method to get $\boldsymbol{R}$ is to let all the clients send their masks $\boldsymbol{r}_i$ to the server for aggregation. However, this leaks information about $\boldsymbol{x}_i$ as the server can easily compute $\boldsymbol{x}_i=\boldsymbol{y}_i-\boldsymbol{r}_i$. 
To deal with this issue, the intuitive idea is to let the clients send partial information of $\boldsymbol{R}$ as $\boldsymbol{R}^i$ rather than $\boldsymbol{r}_i$, which can be used to reconstruct $\boldsymbol{R}$ by computing $\boldsymbol{R} = \mathcal{F}(\boldsymbol{R}^1, \dots, \boldsymbol{R}^n)$ where $\mathcal{F}$ is the reconstruction function. Note that $\boldsymbol{R}^i$ and the process of computing $\boldsymbol{R}^i$ and $\boldsymbol{R}$ should not leak any information of $\boldsymbol{r}_i$. For example, if $\mathcal{F}$ is defined as summation and 
$$\boldsymbol{R}^{i}=\sum_{u_j \in \mathcal{U}} \boldsymbol{r}_{j}^i \text{, where } \boldsymbol{r}_{j}=\sum_{u_i \in \mathcal{U}}\boldsymbol{r}_{j}^i $$
Here, $\boldsymbol{r}_{j}^i$ is the partial information of $\boldsymbol{r}_{j}$ held by client $u_i$. Then, we can have: 
$$\boldsymbol{R}=\sum_{u_j \in \mathcal{U}}\boldsymbol{r}_{j}=\sum_{u_i \in \mathcal{U}}\boldsymbol{R}^{i}$$

Since the client $u_i$ knows only $\boldsymbol{r}_{j}^i$, and server $\mathcal{S}$ knows only $\boldsymbol{R}^{i}$, none of them can learn $\boldsymbol{r}_i$ nor $\boldsymbol{x}_i$. Note that the scheme mentioned above cannot tolerate dropout clients as the server can reconstruct $\boldsymbol{R}$ only when all the clients are connected, i.e., all the $\boldsymbol{R}^{i}$ are received by the server.

\subsection{Handling dropout clients}
\label{sec:handling-dropout-clients}

To handle the scenario where the clients may drop out during aggregation, the server should be able to learn the sum of $\boldsymbol{x}_i$ from only a fraction of the clients. In specific, if there are $n-t$ dropout clients, i.e., $t$ connected users, the server can still reconstruct $\boldsymbol{R}$ from $\boldsymbol{R}^i$ from connected clients. Such a scheme can be achieved by using a $t$-out-of-$n$ Shamir secret sharing scheme. 

Following the description in Section \ref{sec:Shamir-secret-sharing}, we assume a total order on clients. For each client $u_j$, let $\mathcal{F}_{gen}$ be the share generation function and $\mathcal{F}_{rec}$ be the secret reconstruction function for $(t,n)$ standard Shamir secret sharing that $\mathcal{F}_{gen}(\boldsymbol{r}_j) \rightarrow \left\{i,\boldsymbol{r}_j^{i}\right\}_{u_i \in \mathcal{U}}$ where $i \in \{1,2,\dots,n\}$. Here $\boldsymbol{r}_{j}^i$ is the share of $\boldsymbol{r}_{j}$ held by client $u_i$ such that $\mathcal{F}_{rec}(\left\{i,\boldsymbol{r}_j^i\right\}_{u_i \in \mathcal{V}})\rightarrow \boldsymbol{r}_j$ 
where $\left| \mathcal{V}\right| \geq t$. Then, each client $u_i$ can locally compute $\boldsymbol{R}^i$ as the share of $\boldsymbol{R}$ that
$$\boldsymbol{R}^{i}=\sum_{u_j \in \mathcal{U}} \boldsymbol{r}_{j}^{i} \bmod{P}$$
Therefore, relying on the additive homomorphic property of Shamir scheme (see Section \ref{sec:Shamir-secret-sharing}), if the number of received $\boldsymbol{R}^{i}$ is greater than the threshold $t$, the server can reconstruct $  \mathcal{F}_{rec}(\left\{i,\boldsymbol{R}^{i}\right\}_{{u_i} \in \mathcal{V}})\rightarrow\boldsymbol{R}$ where $\left| \mathcal{V}\right| \geq t$.

\subsection{More efficiently generating masks}
\label{sec:generating-masks}

We notice that in the SecAgg scheme \cite{bonawitz2017practical}, the communication cost can be further reduced by having the clients agree on common seeds or keys. In that case, the mask can be generated by a common seed using a pseudorandom generator (PRG). Thus each pair of clients only need to transmit one seed rather than the entire mask. A similar approach can also be applied to our scheme by having each client $u_i$ hold a seed $s_i$ such that the mask vector $\boldsymbol{r}_{i}=\text{PRG}(s_i)$. However, this trick does not work for our scheme as 
$$\text{PRG}(\sum_{u_i \in \mathcal{U}}s_i) \neq \sum_{u_i \in \mathcal{U}}\text{PRG}(s_i)$$
which means
$$\text{PRG}(\sum_{u_i \in \mathcal{U}}s_i) = \boldsymbol{R}  \neq \sum_{u_i \in \mathcal{U}}\boldsymbol{r}_i=\sum_{u_i \in \mathcal{U}}\text{PRG}(s_i)$$
In this case, since $\boldsymbol{y}_i=\boldsymbol{x}_i+\boldsymbol{r}_i$, the server cannot correctly cancel the mask to get the real aggregation results by computing $\boldsymbol{z}=\sum_{u_i \in \mathcal{U}} \boldsymbol{y}_{i}-\boldsymbol{R}$. Luckily, relying on the homomorphic pseudorandom generator (HPRG), this trick still works to reduce the communication costs. As illustrated in Section \ref{sec:homo-prg}, additive homomorphism holds for HPRG such that for any two seeds $s_a$, $s_b \in \mathcal{K}$, we have 
$$\text{HPRG}(s_{a} + s_{b})=\text{HPRG}(s_{a}) \cdot \text{HPRG}(s_{b})$$
Note that HPRG is constructed by using the structure introduced in Section \ref{sec:homo-prg} such that $$\text{HPRG}(k) \rightarrow [F(k,1),F(k,2),F(k,3),...,F(k,m)]$$ 
where $F(k,x)=H(x)^k$ and $m$ is the vector size of $\boldsymbol{x}_i$.
Therefore, in our scheme:
$$\text{HPRG}(\sum_{u_i \in \mathcal{U}}s_i) = \prod_{u_i \in \mathcal{U}}\text{HPRG}(s_i)$$
and component-wisely:
$$\boldsymbol{R} = \prod_{u_i \in \mathcal{U}}\boldsymbol{r}_i = \prod_{u_i \in \mathcal{U}}\text{HPRG}(s_i)$$

To keep the consistency between the operations regarding HPRG and our scheme, we adapt the masking as:
$$\boldsymbol{y}_i=g^{\boldsymbol{x}_i} \cdot \boldsymbol{r}_i$$
Here, $\boldsymbol{r}_i$ is randomly picked from the finite cyclic group $\mathbb{G}$ of which $q$ is the order and $g$ is the generator that all clients agree on. Note that since the distribution of $\boldsymbol{y}_i$ is identical with that of $\boldsymbol{r}_i$, the multiplicative mask $\boldsymbol{r}_{i}$ still guarantees the security. In other words, the mask $\boldsymbol{r}_{i}$ hides all information about $\boldsymbol{x}_i$. Besides, to keep the aggregation results meaningful, the value of $P$ for the Shamir scheme and $q$ for the finite cyclic group $\mathbb{G}$ need to satisfy that $P>q>n*\text{max}(x_i)$ in order to avoid the overflow where $n$ is the number of clients, i.e., the number of vectors to be aggregated. As a result, with the sum of $\boldsymbol{y}_{i}$ and the product of $\boldsymbol{r}_{i}$, we can compute

$$g^{\boldsymbol{z}}=\prod_{u_i \in \mathcal{U}} \boldsymbol{y}_{i}/\boldsymbol{R}=\prod_{u_i \in \mathcal{U}} \boldsymbol{y}_{i}/\prod_{u_i \in \mathcal{U}} \boldsymbol{r}_{i}=g^{\sum_{u_i \in \mathcal{U}}\boldsymbol{x}_i}$$

Because the input range of $x_i$ is fixed which is not very large, computing the discrete logarithms of $g^{\boldsymbol{z}} = g^{\sum_{u_i \in \mathcal{U}}\boldsymbol{x}_i}$ base $g$ to decrypt the sum $\sum_{u_i \in \mathcal{U}}\boldsymbol{x}_i$ is affordable. By using Pollard’s lambda method \cite{menezes2018handbook}, it requires roughly square root of time in the plaintext space. If the plaintext range of each element of $\boldsymbol{x}_i$ is in $\{0, 1, 2, \dots, n\alpha\}$, computing $\boldsymbol{x}_i$ requires roughly $\sqrt{n\alpha}$ time. Such consideration is practical in privacy-preserving aggregation for FL as pointed out in \cite{shi2011privacy}. For example, for 64-bit plaintext space (enough for many image classification tasks), computing a 1024-bit discrete logarithms with $n=500$ for a vector with 50K entries takes about 12.5 seconds on a 64-bit server. Note that in this section, we only describe the high-level overview and omit some details for simplicity, and thus we refer the readers to Section \ref{sec:our-protocol} for the full specification.

\subsection{Putting it all together}
\label{sec:putting-it-all-together}

We summarize the protocols regarding masking as follows:
\begin{enumerate}[1]
    \item Each client $u_i$ randomly selects a seed $s_i \in \mathcal{K}$.
    \item Using $(t, n)$ Shamir secret sharing scheme, each client $u$ computes $n$ shares of $s_i$ as $\mathcal{F}_{gen}(s_i)\rightarrow\left\{j,s_i^j\right\}_{u_j \in \mathcal{U}}$ and then sends $\{j,s_i^j\}$ to the client $u_j \in {\mathcal{U}}$.
    \item  Each client $u_i$ generates $\boldsymbol{r}_i$ based on $\text{HPRG}(s_i)$ and computes the masked update $y_i=g^{\boldsymbol{x}_i} \cdot \boldsymbol{r}_i$ and sends it to the server. Then the server receives the updates from all connected clients (denoted as $\mathcal{V}$), and computes $\prod_{u_i \in \mathcal{V}} \boldsymbol{y}_{i} = g^{\sum_{u_i\in\mathcal{V}}\boldsymbol{x}_i}\cdot \prod_{u_i \in \mathcal{V}}\boldsymbol{r}_i= g^{\sum_{u_i\in\mathcal{V}}\boldsymbol{x}_i}\cdot\boldsymbol{R}$. 
    \item Each client $u_j$ locally computes $s_R^j$ as the share of seed $s_R$ for generating $\boldsymbol{R}$ that $s_R^j=\sum_{u_i \in \mathcal{V}}s_i^j$, and sends $s_R^j$ to the server.
    \item The server reconstructs $s_R$ under $(t, n)$ Shamir secret sharing scheme that $ \mathcal{F}_{rec}(\left\{j,s_R^j\right\}_{u_j \in \mathcal{V}})\rightarrow s_R$, then calculates the mask $\boldsymbol{R}$ using HPRG such that $\boldsymbol{R}=\text{HPRG}(s_R)$.
\end{enumerate}

Note that the list of connected clients $\mathcal{V}$ is kept by the server. The clients need to fetch the list of connected clients before calculating $s_R^j$. We refer readers to Section \ref{sec:our-protocol} for the details. We can observe from the proposed protocol that to mask the clients' model, we rely on the additive property of the Shamir scheme rather than the pair-wise masking as used in SecAgg and SecAgg+. Thus, the communication overheads are significantly reduced. Furthermore, as long as a sufficient number of Shamir shares of the final mask have been collected by the server, the connected clients do not need to send additional Shamir shares to the server, compared to SecAgg-based schemes, which leads to a stronger dropout-resilience with respect to the efficiency.

\subsection{Proposed secure aggregation protocol}
\label{sec:our-protocol}

Our aggregation protocol involves one single server and a set of $n$ clients. Each client $u_i$ has an input vector $\boldsymbol{x}_i$ as its locally trained model or gradient. The vector $\boldsymbol{x}_i$ consists of $m$ elements from field $\mathbb{Z}_q$ for some $q$. The HPRG is under DDH assumption regarding the algorithm $\left(\mathbb{G}, g, q, H\right)$ for HPRG which samples a group $\mathbb{G}$ of order $q$ with generator $g$ and Hash function $H$ (see Section \ref{sec:homo-prg}). Similar to \cite{bonawitz2017practical,bell2020secure}, the communication channels between the server and clients are assumed to be synchronous, which means that message delivery time is bounded, e.g., if the server does not receive an uploaded model from a client within a time limit, it can assume that this client drops out of the system \cite{kairouz2019advances}. 
For simplicity, we assume a public-key crypto-system between any pair of clients, and abuse the notation and use $Enc(msg, pk)$ and $Dec(msg, sk)$ as the encryption and decryption on message $msg$ using the public key $pk$ and the secret key $sk$. The clients may drop out of the protocol at any time.
However, if the number of connected clients is greater than a threshold $t$, the server can still learn the correct output as the aggregation result. 

The detailed description of our protocol is given in Protocol \ref{protocol:our_protocol}.
Specifically, if a client $u_k$ has uploaded the masked model $\boldsymbol{y}_k$ in Step 2 and dropped out in the following steps, $\boldsymbol{y}_k$ still can be unmasked correctly and contribute to the aggregated model as the client has already shared its seed $s_k$ with other clients in Step 1. In other words, when more than $t$ clients send the shares of $\boldsymbol{R}$ to the server, i.e., $|\mathcal{U}_4|>t$, the server can compute an aggregation of all the models uploaded by clients in $U_2$ ($U_3$ under malicious threat model).
Note that if the client $u_k$ has shared the seed $s_k$ successfully but drops out before sending its masked model, i.e., $u_k \in \mathcal{U}_1 \setminus \mathcal{U}_2$, the mask vector $\boldsymbol{r}_k$ of the client $u_k$ will not be included in $\boldsymbol{R}$. As shown in Step 4, each connected client only computes the sum of shares of the seeds getting from the clients in $\mathcal{U}_2$. In this scenario, the share of client $u_k$'s seed $s_k$ is not added and thus $\boldsymbol{r}_k$ will not be reconstructed by the server.
Like other mechanisms based on Shamir secret sharing \cite{bonawitz2017practical, bell2020secure}, to reconstruct secrets, the proposed protocol is inevitably subject to delay since clients have to wait for the $\mathcal{U}_2$ list, which contains the identifier of the clients from whom the server has received the masked models. Nevertheless, as the protocol is assumed to run on a synchronous channel as mentioned earlier, this delay is bounded and the client is considered dropped out of the protocol if it does not upload the masked model within a predefined time limit.

We can observe that compared with the SecAgg scheme, since the most communication-intensive building block, i.e., the pair-wise DH key agreement protocol for seed agreement, is removed, our scheme achieves a much simpler structure.

\setcounter{figure}{0}
\begin{figure*}
\centering
\renewcommand\figurename{Protocol}
\begin{protocol}
\begin{adjustbox}{minipage=1.95\linewidth,fbox={\fboxrule} 12pt}
\vspace{-0.1cm}
\centerline{\bfseries{Protocol: HPRG based Secure Aggregation\par}}~
\vspace{-0.2cm}\\
\textbf{{Participants:}} A single central server $\mathcal{S}$ and a set of clients $\mathcal{U}$.\\
\textbf{{Private inputs:}} Each client $u_i$ has a locally trained model or gradient represented as a vector $\boldsymbol{x}_i$, a secret key for constructing authenticated channels $csk_i$,  \underline{and a secret key for signature $ssk_i$.} The server has a secret key for authenticated channels with clients $csk_s$  \underline{ and a secret key for signature $ssk_s$}.\\
\textbf{{Public inputs:}} The number of clients $n=\left|\mathcal{U}\right|$, the threshold $t<n$, the field $\mathbb{Z}_P$ for some $P$ for Shamir secret sharing scheme with function $\mathcal{F}_{gen}$ and $\mathcal{F}_{rec}$,
the algorithm $(\mathbb{G}, g, q)$ for HPRG which samples a finite cyclic group $\mathbb{G}$ of prime order $q$ with generator $g$, and the security parameter $\kappa$. Each client $u_i$'s and server's public key for constructing secure channels $cpk_i$, $cpk_s$,  \underline{and their public keys for signature $spk_i$, $spk_s$}. Note that $P>q>n*\text{max}(x_i)$.\\
\textbf{{Outputs:}} The aggregation result of the locally trained models from the set of connected clients $\mathcal{U}_2 \subseteq \mathcal{U}$: $\sum_{u_i \in \mathcal{U}_2}\boldsymbol{x}_i$  \underline{($\mathcal{U}_3$ under malicious threat model)}
\vspace{0.1cm}
\hrule
\hrule
\vspace{0.0cm}
\begin{enumerate}[*]
    \item \textbf{{Step 1 - Sharing seeds:}}\\
    Client $u_i$:
    \begin{enumerate}
        \item The client randomly picks $s_i$ from $\mathbb{Z}_q$ where $q$ is the order of the finite cyclic group $\mathbb{G}$ from which all clients agree on a generator $g$. Generates $(t, n)$ Shamir secret shares of $s_i \in \mathbb{Z}_P$, i.e., $\mathcal{F}_{gen}(s_i)\rightarrow (\left\{j,{s}_{i}^j\right\}_{u_j \in \mathcal{U}})$
        \item Sends $(Enc({s}_{i}^j, cpk_j),  \underline{\sigma_{i,j}^1})$ to client $u_j \in \mathcal{U}$,  \underline{where the signature $ \mathcal{F}_{sig}(Enc({s}_{i}^j, cpk_j), ssk_i)\rightarrow\sigma_{i,j}^1$}  (Denote $\mathcal{U}_1$ as the set of clients $u_i\in \mathcal{U}$ that at least $t$ shares of $s_i$ have been received by other clients).
        \item Receives $(Enc({s}_{j}^i, cpk_i),  \underline{\sigma_{j,i}^1})$ from all the clients $u_j \in \mathcal{U}$, and then computes ${s}_{j}^u=Dec(Enc({s}_{j}^u, cpk_i), csk_i)$.  \underline{If $\mathcal{F}_{vrfy}(Enc({s}_{j}^i, cpk_i), spk_j,  \sigma_{j,i}^1)=0$, aborts.}
    \end{enumerate}

    \item \textbf{{Step 2 - Collecting masked models:}}\\
    Client $u_i$:
    \begin{enumerate}
    \item Generates the mask vector $\boldsymbol{r}_i$ using HPRG with the seed $s_i$, i.e., $\boldsymbol{r}_i=\text{HPRG}(s_i)$.
    \item Computes the model masked by $\boldsymbol{r}_i$, i.e., $\boldsymbol{y}_i=g^{{\boldsymbol{x}_i}} \cdot \boldsymbol{r}_i$.
    \item Sends $\boldsymbol{y}_i$  \underline{with the signature $\mathcal{F}_{sig}(\boldsymbol{y}_i, ssk_i)\rightarrow\sigma^2_i$} to the server $\mathcal{S}$ .
    \end{enumerate}
    Server $\mathcal{S}$:
    \begin{enumerate}
    \item Receives all $\boldsymbol{y}_i$  \underline{with the signature $\sigma^2_i$} from clients (denote with $\mathcal{U}_2 \subseteq \mathcal{U}_1$ this set of clients).  \underline{If $\mathcal{F}_{vrfy}(\boldsymbol{y}_i,$}  \underline{ $spk_i, \sigma^2_i)=0$, remove client $u_i$ from $\mathcal{U}_2$.}
    \end{enumerate}

    \item \textbf{{ \underline{Step 3 - Checking consistency}}:}\\
    Client $u_i$:
    
    \begin{enumerate}
        \item  \underline{Fetches the list of $\mathcal{U}_2$ from the server $\mathcal{S}$ with the server's signature $\mathcal{F}_{sig}(\mathcal{U}_2, ssk_s)\rightarrow\sigma^3_s$. If $\mathcal{F}_{vrfy}(\mathcal{U}_2, $}  \underline{ $spk_s, \sigma^3_s)=0$, aborts.}
        \item  \underline{Sends $\mathcal{F}_{sig}(\mathcal{U}_2, ssk_i)\rightarrow\sigma^4_i$ to the server $\mathcal{S}$.}
    \end{enumerate}
    Server $\mathcal{S}$:
    \begin{enumerate}
    \item  \underline{Receives $\sigma^4_i$ from at least $t$ clients (Denote with $\mathcal{U}_3 \subseteq \mathcal{U}_2$ this set of clients) and forwards to the clients in $\mathcal{U}_3$.}
    \end{enumerate}
    
    \item \textbf{Step 4 - Unmasking:}\\
    Client $u$:
    \begin{enumerate}
        \item If the protocol does not consist of step 3 for consistency checking, fetches the list of $\mathcal{U}_2$ from the server $\mathcal{S}$.  \underline{ Otherwise, if $|\mathcal{U}_3|<t$ or for all $u_j \in \mathcal{U}_3$, $\mathcal{F}_{vrfy}(\mathcal{U}_2, spk_j, \sigma^4_j)=0$, aborts. }
        \item Computes the sum of shares of $s_j^i$ from all the clients $u_j \in \mathcal{U}_2$, i.e., $s_R^i=\sum_{u_j \in \mathcal{U}_2}s_j^i \bmod{P}$.  \underline{($\mathcal{U}_3$ under malicious threat model)}
        \item Sends $s_R^i$  \underline{with the signature $\mathcal{F}_{sig}(Enc(s_R^i, cpk_s), ssk_i)\rightarrow\sigma_j^5$} to the server $\mathcal{S}$.
    \end{enumerate}
    Server $\mathcal{S}$:
    \begin{enumerate}
    \item Receives $s_R^j$ from the clients (Denote with $\mathcal{U}_4$ this set of clients).  \underline{ If $\mathcal{F}_{vrfy}(Enc(s_R^j, cpk_s), spk_j, \sigma^5_j)=0$, }  \underline{ remove client $j$ from $\mathcal{U}_{4}$.} Proceed until $\left|\mathcal{U}_{4}\right|>t$. 
    \item Reconstructs the seed for unmasking, i.e., $\mathcal{F}_{rec}(\left\{j,s_R^j\right\}_{u_j \in \mathcal{U}_4}) \rightarrow s_R  \bmod{P}$.
    \item Generates the mask $\boldsymbol{R}$ using HPRG with the seed $s_R$, i.e., $\boldsymbol{R}=\text{HPRG}(s_R)$.
    \item Computes $g^{\boldsymbol{z}}=\prod_{u \in \mathcal{U}_2} \boldsymbol{y}_{i}/\boldsymbol{R}=\prod_{u_i \in \mathcal{U}_2} \boldsymbol{y}_{i}/\prod_{u_i \in \mathcal{U}_2} \boldsymbol{r}_{i}=g^{\sum_{u_i \in \mathcal{U}_2}\boldsymbol{x}_i}$.  \underline{($\mathcal{U}_3$ under malicious threat model)}
    \item Computes and outputs $\boldsymbol{z}=\log_{g}({g}^{\boldsymbol{z}})=\sum_{u_i \in \mathcal{U}_2}\boldsymbol{x}_i$.  \underline{($\mathcal{U}_3$ under malicious threat model)}
    \end{enumerate}

\end{enumerate}
\end{adjustbox}
\end{protocol}
\vspace{-0.5cm}
\caption{The description of HPRG based Secure Aggregation protocol for one FL round.  \underline{The underlined parts are only for active malicious threat model.}}
\label{protocol:our_protocol}
\end{figure*}

\section{Security Analysis}
\label{sec:sec-analysis}

In this section, we provide the security claims along with their proofs for the protocols proposed in Section~\ref{sec:our-protocol}. Recall that the involved participants are a single central server $\mathcal{S}$ and a set of clients $\mathcal{U}$ with their locally trained models $x_{\mathcal{U}}$. We consider that the central server is always online while the clients may abort, e.g., drop out, from the protocol at any point. We denote the connected clients in each step as $\mathcal{U}_i$ from the receivers' angle (refer to Protocol \ref{protocol:our_protocol}). The underlying cryptographic building blocks are instantiated with the security parameter $\kappa$. 

We assume a group of adversaries consisting of a subset of clients whose number is less than a threshold $t$, and with or without the central server. The security definition requires that any group of adversaries will learn nothing about the remaining clients' values. For example, if the threshold $t>2$, a group of adversaries consisting of client $u_1, u_2$ and the central server $\mathcal{S}$, they will not learn any information about $ u_3$'s locally trained model. More specifically, given any subset $\mathcal{C} \subseteq \mathcal{U}$ of the adversaries where $\left|\mathcal{C} \setminus \{\mathcal{S}\} \right |<t$, the resulting values of honest participants $\mathcal{U} \setminus \mathcal{C}$ should look uniformly random.

Let $\mathsf{REAL}_\mathcal{C}^{\mathcal{U}, t, \kappa}(x_{\mathcal{U}},\mathcal{U}_1, \mathcal{U}_2, \mathcal{U}_3)$ be a random variable representing the joint views of participants in $\mathcal{C}$ in real execution of our proposed protocol, and $\mathsf{SIM}_\mathcal{C}^{\mathcal{U}, t, \kappa}(x_{\mathcal{U}},\mathcal{U}_1, \mathcal{U}_2, \mathcal{U}_3)$ be another combined views of participants in $\mathcal{C}$ simulating the protocol that the inputs of honest participants are selected randomly and uniformly denoted with $x_\mathcal{C}$. Following above-mentioned idea, the distribution of $\mathsf{REAL}_\mathcal{C}^{\mathcal{U}, t, \kappa}(x_{\mathcal{U}},\mathcal{U}_1, \mathcal{U}_2,\mathcal{U}_3)$ and $\mathsf{SIM}_\mathcal{C}^{\mathcal{U}, t, \kappa}(x_{\mathcal{C}},\mathcal{U}_1, \mathcal{U}_2, \mathcal{U}_3)$ should be indistinguishable.

\subsection{Semi-honest Model}
We consider two cases: (i) a subset of clients are semi-honest colluding adversaries, and the server is honest (ii) the server is additionally semi-honest adversarial and colludes with a subset of clients. We provide the security claims along with the proofs in Theorem \ref{sec:semi-honest-1} and Theorem \ref{sec:semi-honest-2} respectively.

\begin{theorem}[Security against semi-honest clients, with honest server]
For all $\mathcal{U}, t, \kappa$ with $\left|\mathcal{C} \right|<t$, $x_{\mathcal{U}},\mathcal{U}_1, \mathcal{U}_2,\mathcal{U}_3$, and $\mathcal{C}$ where $\mathcal{C} \subseteq \mathcal{U}$ and $\mathcal{U}_3 \subseteq \mathcal{U}_2 \subseteq \mathcal{U}_1 \subseteq \mathcal{U}$, there exists a probabilistic polynomial-time (PPT) simulator $\mathsf{SIM}$ such that
$$\mathsf{SIM}_\mathcal{C}^{\mathcal{U}, t, \kappa}(x_{\mathcal{C}},\mathcal{U}_1, \mathcal{U}_2, \mathcal{U}_3) \equiv \mathsf{REAL}_\mathcal{C}^{\mathcal{U}, t, \kappa}(x_{\mathcal{U}},\mathcal{U}_1, \mathcal{U}_2,\mathcal{U}_3)$$
where $\equiv$ denotes that the distributions are identical.
\label{sec:semi-honest-1}
\end{theorem}
\begin{proof}
Since the server is honest, the combined views of the participants in $\mathcal{C}$ are independent of that of the participants who are not in $\mathcal{C}$. This means by letting all semi-honest participants have their actual inputs and other participants have dummy inputs, the $\mathsf{SIM}$ can perfectly simulate the views of the participants in $\mathcal{C}$. As only the list of specific participants will be revealed to semi-honest participants, the simulator $\mathsf{SIM}$ can set the message uploaded from honest participants who are not in $\mathcal{C}$ as dummy values. Therefore, the simulated combined views of the participants in $\mathcal{C}$ are identical to that in $\mathsf{REAL}$.
\end{proof}

\begin{theorem}[Security against semi-honest adversaries, including the server]
For all $\mathcal{U}, t, \kappa$ with $\left|\mathcal{C}\setminus\{\mathcal{S}\} \right|<t$, $x_{\mathcal{U}},\mathcal{U}_1, \mathcal{U}_2,\mathcal{U}_3$, and $\mathcal{C}$ where $\mathcal{C} \subseteq \mathcal{U} \cup \{\mathcal{S}\}$ and $\mathcal{U}_3 \subseteq \mathcal{U}_2 \subseteq \mathcal{U}_1 \subseteq \mathcal{U}$, there exists a probabilistic polynomial-time (PPT) simulator $\mathsf{SIM}$ such that
$$\mathsf{SIM}_\mathcal{C}^{\mathcal{U}, t, \kappa}(x_{\mathcal{C}},\mathcal{U}_1, \mathcal{U}_2, \mathcal{U}_3) \equiv \mathsf{REAL}_\mathcal{C}^{\mathcal{U}, t, \kappa}(x_{\mathcal{U}},\mathcal{U}_1, \mathcal{U}_2,\mathcal{U}_3)$$
where $\equiv$ denotes that the distributions are identical.
\label{sec:semi-honest-2}
\end{theorem}
\begin{proof}
We use a standard hybrid argument to prove the theorem. We define a sequence of hybrid distributions $H_0, H_1, \dots$ to construct the simulator $\mathsf{SIM}$ by the subsequent modifications to the random variable $\mathsf{REAL}$. In other words, if any two subsequent hybrids are computationally indistinguishable, the distribution of simulator $\mathsf{SIM}$ as a whole is also identical to the real execution $\mathsf{REAL}$.
\begin{enumerate}[*]
    \item $\mathsf{H}_0$: In this hybrid, the distribution of the combined views of $\mathcal{C}$ of $\mathsf{SIM}$ is exactly the same as that of $\mathsf{REAL}$.
    \item $\mathsf{H}_1$: In this hybrid, for each client $u_i \in \mathcal{U}_1 \setminus \mathcal{C}$, we replace ${s}_{i}^j$, i.e., the share of simulated honest client $u_i$'s seed $s_i$ that sent to adversarial client $u_j$, with a randomly selected element in the corresponding field. Note that since the adversaries in $\mathcal{C}$ do not receive any additional shares of $s_i$ where $u_i \in \mathcal{U}_1 \setminus \mathcal{C}$, the combined view of adversaries has only $\left|\mathcal{C} \setminus \{\mathcal{S}\} \right| <t$  shares of each seed $s_i$. According to the security property of the Shamir secret sharing scheme, the adversaries learn nothing about the seed $s_i$. Therefore, the distribution of this hybrid is identical to the previous one.
    \item $\mathsf{H}_2$: In this hybrid, instead of computing the mask $\boldsymbol{r}_i$ using HPRG, the mask $\boldsymbol{r}_i$ of each simulated client $u_i$ is replaced with a randomly selected number in the appropriate length. Since in the previous hybrid, the seed $s_i$ is chosen uniformly and randomly by letting its shares be selected uniformly at random, and the adversaries' seeds are set to be 0, the output of HPRG does not depend on its seed. Therefore, the security of HPRG leveraging the Decisional Diffie-Hellman assumption guarantees the identical distribution of this hybrid to the previous one.
    \item $\mathsf{H}_3$: In this hybrid, we substitute each of the honest clients $u_i$'s masked model $\boldsymbol{y}_i$ with a uniformly random value. Since in the previous hybrid, $\boldsymbol{r}_i$ is chosen uniformly at random and is used as a one-time pad to mask $\boldsymbol{x}_i$, it is obvious that $\mathsf{SIM}$ can simulate $\mathsf{REAL}$ without knowing any information about $\boldsymbol{x}_i$, and thus this hybrid is identically distributed to the previous one.
    \item $\mathsf{H}_4$: In this hybrid, for the honest clients $u_i \in \mathcal{U}_2 \setminus \mathcal{C}$, instead of sending $g^{\boldsymbol{x}_i} \cdot \boldsymbol{r}_i \rightarrow\boldsymbol{y}_i$, we send $g^{\boldsymbol{w}_i} \cdot \boldsymbol{r}_i \rightarrow\boldsymbol{y}_i $ where $\boldsymbol{w}_i$ is uniformly sampled at random from $\mathbb{Z}_q$, subject to $$\sum_{u_i \in \mathcal{U}_2 \setminus \mathcal{C}} \boldsymbol{w}_i = \sum_{i \in \mathcal{U}_2 \setminus \mathcal{C}} \boldsymbol{x}_i \bmod{q}$$
    We can observe that the distribution of $g^{\boldsymbol{x}_i}\cdot \boldsymbol{r}_i $ is identical to that of $g^{\boldsymbol{w}_i} \cdot \boldsymbol{r}_i$ subject to the above equation.
\end{enumerate}
By defining such PPT simulator $\mathsf{SIM}$ as described in the last hybrid, the semi-honest adversaries' combined views of $\mathsf{SIM}$ are computationally indistinguishable from that of the real execution $\mathsf{REAL}$, and thus the proof is completed.
\end{proof}

\subsection{Active Malicious Model}
\label{sec:malicious-model}

Next, we discuss the security of the active malicious threat model. Note that in such a threat model, correctness cannot be guaranteed as active adversaries $\mathcal{C}$ can deviate from the protocol at any time by sending fraudulent messages, distorting the outputs, etc. In this case, only the privacy of honest clients' inputs is considered to be guaranteed. The difference between the semi-honest model and malicious model for our protocol can be summarized as follows:
\begin{enumerate}[-]
    \item The adversaries $\mathcal{C}$ can simulate a specific honest client $u$, and thus receive all the related information about $u$ to recover $u$'s inputs. This malicious behavior is so-called the Sybil attack.
    \item The malicious server $\mathcal{S}$ can actively send a different list of connected clients to the honest clients. For example, the server sends $\mathcal{U}$ to client $u$, and $\mathcal{V} \subseteq \mathcal{U}$ to client $v$. Then during the execution of the protocol, the information about $\mathcal{U} \setminus \mathcal{V}$ would be leaked, which can be used to recover the private inputs of honest clients' from $\mathcal{U} \setminus \mathcal{V}$.
    \item The malicious participants (server) is able to dynamically set any honest client to be dropped out from the protocol at any round of protocol execution, and thus the proof for the semi-honest threat model is no longer correct. The reason is that the simulator $\mathsf{SIM}$ knows only the sum of honest clients' inputs. If some honest clients are set to be dropped, the $\mathsf{SIM}$ cannot simulate the rest of honest clients' behaviors as it knows nothing about the private inputs of dropped clients.
\end{enumerate}

The first difference regarding the Sybil attack can be solved by using a standard signature scheme (see Section \ref{sec:signature}) which can be used to prove the origin of a message. In specific, a message signed by client $u$ must have come from $u$, and its origin can be verified by any participants. Thus the messages sent from honest clients cannot be modified or substituted by malicious adversaries.

However, even with authenticated channels, the malicious server can still give a different view of dropped clients (connected clients) to the honest clients for malicious purposes, as pointed out in the second difference. Therefore, we have to check the consistency between the list of connected clients sent to each client. In general, after receiving the list of the connected clients $\mathcal{V}$ from the server, each client $u_i$ in $\mathcal{V}$ generates a signature $\sigma_i$ on $\mathcal{V}$ and sends it to the server. Then the server forwards all the $\sigma_i$ to the clients in $\mathcal{V}$ to have them checking the consistency between $\mathcal{V}$ and $\{\sigma_i\}_{u_i \in \mathcal{V}}$, i.e., to verify if each $\sigma_i$ is actually the client $u_i$'s signature on $\mathcal{V}$ (see Step 3 in Protocol \ref{protocol:our_protocol}). As a result, the same view of connected clients lists to all the clients in $\mathcal{V}$ can be guaranteed. Note that the consistency check costs one communication round, which may cause more dropped clients, and thus needs to maintain an extra list of connected clients.

For the last difference, we take a similar approach in \cite{bonawitz2017practical} to adopt the proof to be performed in random oracle (RO). In such an RO model, the simulator $\mathsf{SIM}$ is able to send a query to an ideal functionality to learn the sum of a dynamically selected subset of honest clients. In other words, by reprogramming the RO such that the subset of honest clients is chosen dynamically, the combined view of adversaries in the real protocol execution $\mathsf{REAL(M_C)}$ is indistinguishable from that of the simulator $\mathsf{SIM}$. Here $M_C$ is a probabilistic polynomial-time algorithm that denotes the ``next message" function of participants in $\mathcal{C}$, which enables participants in $\mathcal{C}$ to dynamically choose (i) their inputs at any round of the protocol execution and (ii) the list of connected participants. 

We first give an ideal functionality $\mathcal{F}_{HSecAgg}$ in Func. 1 that describes how a fully trusted third party $\mathcal{T}$ would compute each participant's output from the inputs, i.e., calculate the sum of each client's model. In this case, our proposed protocol is secure if a simulator can simulate any information that the malicious adversaries can learn from the protocol in such a way that it is indistinguishable from what they can learn from the ideal functionality $\mathcal{F}_{HSecAgg}$, i.e., $\mathsf{SIM(\cdot)}\equiv \mathsf{REAL(\cdot)}$. To enable the comparison between our scheme and previous schemes, including SecAgg and SecAg+, we provide similar security claims under two settings, i.e., active malicious clients with and without honest server, along with the proofs in Theorem \ref{sec:malicious-1} and Theorem \ref{sec:malicious-2}, respectively.

\begin{function}
\caption{Functionality $\mathcal{F}_{HSecAgg}$. HPRG based Secure Aggregation Scheme Overview.}
\begin{algorithmic}[1]

\renewcommand{\algorithmicrequire}{\textbf{Participants:}}
\REQUIRE ~~\\
$\bullet$ A single central server $\mathcal{S}$ and a set of clients $\mathcal{U}$.\\
\renewcommand{\algorithmicrequire}{\textbf{Inputs:}}
\REQUIRE ~~\\
$\bullet$ Private gradient vector of each client $\boldsymbol{x}_i$;
Private seed of each client $s_i$;
Public number of clients $n=|\mathcal{U}|$;
Public $(t,n)$ Shamir secret sharing scheme.\\
\renewcommand{\algorithmicensure}{\textbf{Outputs:}}
\ENSURE ~~\\
$\bullet$ The server receives the aggregation result of the gradient vectors from a of clients $\mathcal{U}_3 \subseteq \mathcal{U}$.\\
~~\\

Trusted party $\mathcal{T}$ executes the following steps:

\STATE $\mathcal{T}$ receives the seeds from a set of clients $\mathcal{U}_1 \subseteq \mathcal{U}$. If $|\mathcal{U}_1|<t$, aborts.
\STATE $\mathcal{T}$ computes and sends the Shamir shares of received seeds and the list of $\mathcal{U}_1$ to the clients in $\mathcal{U}_1$.
\STATE $\mathcal{T}$ receives the masked models $\boldsymbol{y}_i$, which are constructed based on seeds' Shamir shares and $\boldsymbol{x}_i$, from a set of clients $\mathcal{U}_2 \subseteq \mathcal{U}_1$. If $|\mathcal{U}_2|<t$, aborts.
\STATE $\mathcal{T}$ sends the list of $\mathcal{U}_3 \subseteq \mathcal{U}_2$ to the clients in $\mathcal{U}_2$. If $|\mathcal{U}_3|<t$, aborts. (This step involves consistency check.)
\STATE $\mathcal{T}$ receives Shamir shares constructed based on seeds' Shamir shares for unmasking from a set of clients $\mathcal{U}_4 \subseteq \mathcal{U}_3$. If $|\mathcal{U}_4|<t$, aborts.
\STATE $\mathcal{T}$ calculates the sum of $\boldsymbol{x}_i$ for $\mathcal{U}_3$ based on the received $\boldsymbol{y}_i$ from $\mathcal{U}_3$ and the received Shamir shares from $\mathcal{U}_4$, then send it to the server.
\label{alg:ideal_func}

\end{algorithmic}
\end{function}

\begin{theorem}[Security against active malicious clients, with honest server]
For all $\mathcal{U}, t, \kappa$ with $\left|\mathcal{C} \right|<t$, $x_{\mathcal{U\setminus C}}$, and $\mathcal{C}$ with the algorithm $M_C$ where $\mathcal{C} \subseteq \mathcal{U}$, the protocol \ref{protocol:our_protocol} is a secure protocol for computing $\mathcal{F}_{HSecAgg}$, i.e., there exists a probabilistic polynomial-time (PPT) simulator $\mathsf{SIM}$ such that
$$\mathsf{SIM}_\mathcal{C}^{\mathcal{U}, t, \kappa}(M_C, x_{\mathcal{U\setminus C}}) \equiv \mathsf{REAL}_\mathcal{C}^{\mathcal{U}, t, \kappa}(M_C)$$
where $\equiv$ denotes that the distributions are identical.
\label{sec:malicious-1}
\end{theorem}
\begin{proof}
The proof is identical to that for Theorem \ref{sec:semi-honest-1}. The reason is that even with $M_C$, which enables participants in $\mathcal{C}$ to choose their inputs at any round of the protocol execution, participants in $\mathcal{C}$ learn nothing about $x_{\mathcal{U\setminus C}}$ rather than the list of connected participants. Therefore, the simulator $\mathsf{SIM}$ can let all active malicious participants have their actual inputs and other participants have dummy inputs to perfectly simulate the views of the participants in $\mathcal{C}$. In this case, the simulated combined views of the participants in $\mathcal{C}$ are identical to that in $\mathsf{REAL}$.
\end{proof}

However, for the threat model including both active malicious clients and server, the proof is different from that for Theorem \ref{sec:semi-honest-2} as the sum of $\boldsymbol{x}_i$ is no longer available as the input to the simulator $\mathsf{SIM}$ (see $\mathsf{H}_4$ in the proof for Theorem \ref{sec:semi-honest-2}). Thus, we allow $\mathsf{SIM}$ to learn the sum by making a query to an ideal function $I$ in RO for dynamically chosen subset of honest participants denoted as $\mathcal{L}$ at any round of the protocol execution. More precisely, the ideal function $I$ takes $\mathcal{L},\mathcal{U},\mathcal{C}$ and a lower bound of the number of honest participants $\delta$ as inputs, and outputs $\sum_{u_i \in \mathcal{L}} \boldsymbol{x}_i$ if $\mathcal{L} \subseteq (\mathcal{U}-\mathcal{C}) $ and $|\mathcal{L}|\geq \delta$, and aborts otherwise.

\begin{theorem}[Security against active malicious clients, including the server]
For all $\mathcal{U}, t, \kappa$ with $\left|\mathcal{C}\setminus\{\mathcal{S}\} \right|<t$, $x_{\mathcal{U\setminus C}}$, $\mathcal{C}$ with the algorithm $M_C$ where $\mathcal{C} \subseteq \mathcal{U} \cup \{\mathcal{S}\}$, and $\delta=t-|\mathcal{C} \cap \mathcal{U}|$, the protocol \ref{protocol:our_protocol} is a secure protocol for computing $\mathcal{F}_{HSecAgg}$, i.e., there exists a probabilistic polynomial-time (PPT) simulator $\mathsf{SIM}$ such that
$$\mathsf{SIM}_\mathcal{C}^{\mathcal{U}, t, \kappa}(M_C, x_{\mathcal{U\setminus C}}) \equiv \mathsf{REAL}_\mathcal{C}^{\mathcal{U}, t, \kappa, I}(M_C)$$
where $\equiv$ denotes that the distributions are identical.
\end{theorem}
\label{sec:malicious-2}
\begin{proof}
Similar to the proof for Theorem \ref{sec:semi-honest-2}, we use a standard hybrid argument to prove the theorem. By defining a sequence of modifications to the random variable $\mathsf{REAL}$, we can construct the simulator $\mathsf{SIM}$ with a sequence of hybrid distributions. If any two subsequent hybrids are computationally indistinguishable, the distribution of simulator $\mathsf{SIM}$ as a whole is identical to the real execution $\mathsf{REAL}$.

\begin{enumerate}[*]
    \item $\mathsf{H}_0$: In this hybrid, the distribution of the combined views of $M_C$ of $\mathsf{SIM}$ is exactly the same as that of $\mathsf{REAL}$.
    \item $\mathsf{H}_1$: In this hybrid, we substitute all the shares ${s}_{i}^j$ for each client $u_i \in \mathcal{U}_1 \setminus \mathcal{C}$, with a randomly selected element in $\mathbb{Z}_P$. The security of Shamir secret sharing scheme guarantees identical distribution from the previous one.
    \item $\mathsf{H}_2$: In addition to the previous hybrid, the $\mathsf{SIM}$ aborts if $M_C$ provides any incorrect $\sigma_{i,j}^1$. Since this is equivalent to breaking the security of the signature scheme, this hybrid is identical to the previous one.
    \item $\mathsf{H}_3$: In this hybrid, the mask $\boldsymbol{r}_i$ of each simulated client $u_i$ is substituted with a randomly selected number in appropriate length, and the adversaries' masks are set to be 0. The security of HPRG leveraging the Decisional Diffie-Hellman assumption guarantees the identical distribution of this hybrid to the previous one.
    \item $\mathsf{H}_4$: In addition to the previous hybrid, the $\mathsf{SIM}$ aborts if $M_C$ provides any incorrect $\sigma^2_i$. Since this is equivalent to breaking the security of the signature scheme, this hybrid is identical to the previous one.
    \item $\mathsf{H}_5$: In addition to the previous hybrid, the $\mathsf{SIM}$ aborts if $M_C$ provides any incorrect $\sigma^3_s$. Because of the security of the signature scheme that guarantees that forgeries can happen only with negligible probability, this hybrid is indistinguishable from the previous one.
    \item $\mathsf{H}_6$: In addition to the previous hybrid, the $\mathsf{SIM}$ aborts if $M_C$ provides any incorrect $\sigma^4_i$. Because of the security of the signature scheme that guarantees that forgeries can happen only with negligible probability, this hybrid is indistinguishable from the previous one.
    \item $\mathsf{H}_7$: Denote the list of $\mathcal{U}_2$ fetched from the server as $\mathcal{Q}$. The $\mathsf{SIM}$ aborts if two different $\mathcal{Q}$ are signed by the clients. Since this amounts to breaking the security discussed in Section \ref{bound_malicious_clients_and_server}, and the server cannot forge signatures on behalf of the honest clients, this hybrid is indistinguishable from the previous one.
    \item $\mathsf{H}_8$: In addition to the previous hybrid, the $\mathsf{SIM}$ aborts if $M_C$ provides any incorrect $\sigma^4_i$. Because of the security of the signature scheme that guarantees that forgeries can happen only with negligible probability, this hybrid is indistinguishable from the previous one.
    \item $\mathsf{H}_9$: In this hybrid, the simulator $\mathsf{SIM}$ does not receive the inputs of the honest participants. Instead, it learns the required value $\boldsymbol{w}_i$ for the set $(\mathcal{Q} \setminus \mathcal{C})$ with respect to $\sum_{u_i \in \mathcal{Q} \setminus \mathcal{C}} \boldsymbol{w}_i = \sum_{u_i \in \mathcal{Q} \setminus \mathcal{C}} \boldsymbol{x}_i$ by making a query to the ideal function $I$. Note that according to the discussion in $\mathsf{H}_7$ and Section \ref{bound_malicious_clients_and_server}, the ideal function $I$ will not abort and does not modify the joint view of $\mathcal{C}$. Thus, this hybrid is indistinguishable from the previous one.
    
\end{enumerate}
By defining such PPT simulator $\mathsf{SIM}$ as described in the last hybrid, the active malicious adversaries' combined views of $\mathsf{SIM}$ is computationally indistinguishable from that of the real execution $\mathsf{REAL}$, and thus the proof is completed. This means the active malicious participants learn nothing except for the sum of $\boldsymbol{x}_i$ where $u_i \in \mathcal{L}$ and $|\mathcal{L}|\geq \delta$.
\end{proof}

Note that for active adversaries, the security of the consistency check can be guaranteed only when the threshold $t$ with respect to $\delta$ is set to be a proper value. Next, we discuss the minimum value of $t$ required for security in three threat models as follows.

\subsubsection{Malicious clients with honest server}
It is evident that with an honest server, each honest client views the correct list of connected clients, which is guaranteed by the protocol of consistency check. This means that for any $t>1$, the client learns nothing about the information of other's inputs.

\subsubsection{Honest clients with malicious server} Compared with the previous model, malicious server makes the situation a bit more complicated. The reason is that for some specific threshold $t$, the server is able to learn the information of the inputs, while passing the consistency check. For example, there are four clients $u_1, u_2, u_3, u_4$ with threshold $t=2$, the malicious server actively sends the connected client list $l_1=\{u_1, u_2, u_3\}$ to $u_1, u_2$, and $l_2=\{u_1, u_2, u_3, u_4\}$ to $u_3, u_4$, then following the consistency check protocol, $u_1$ and $u_2$ receive $l_1$ with their signatures on $l_1$, say $\sigma_1\{l_1\}$ and $\sigma_2\{l_1\}$. Since they can verify the origin of the messages and $|\{\sigma_1\{l_1\}, \sigma_2\{l_1\}\}| =2 \geq t$, $u_1$ and $u_2$ do not abort during the consistency check, and thus $u_3$ and $u_4$ are viewed as dropped clients from the angle of $u_1$ and $u_2$. Consequently, in our protocol, the server will receive the product of masks from connected clients $R_1=r_1r_2r_3$ where $r_i$ is the mask generated by $u_i$. Similarly, the server will also receive $R_2=r_1r_2r_3r_4$. As a result, $r_4$ can be easily obtained by computing $\frac{R_1}{R_2}$ which causes leakage of $u_4$'s input. To avoid such different views for $n$ clients, the threshold $t$ has to be set to be greater than $\frac{n}{2}$. In this case, since all the clients are honest, no one will make the signature twice, the number of signatures equals to $n$. This means that for two fraudulent list $l_1$ and $l_2$, if $l_1$ passes the consistency check with $|\{\sigma\{l_1\}\}|>t$, the list $l_2$ cannot pass the consistency check as $|\{\sigma\{l_2\}\}|=n-|\{\sigma\{l_1\}\}|<t$. Therefore, the security is guaranteed for $t \geq \lfloor  \frac{n}{2} \rfloor +1$.

\subsubsection{Malicious clients and server}
\label{bound_malicious_clients_and_server}
Different from the assumption in the previous model, for malicious clients, they are able to make the signature on the list $l$ any number of time. For example, there are six clients $u_1, u_2, u_3, u_4, u_5, u_6$ with the threshold $t=4$ where $u_1$ and $u_2$ are active malicious clients. The client $u_3$ and $u_4$ receive the list $l_1=\{u_1, u_2, u_3, u_4, u_5\}$ with the signatures $\sigma_1=\{\sigma_1\{l_1\},\sigma_2\{l_1\}, \sigma_3\{l_1\}, \sigma_4\{l_1\}\}$ while the client $u_5$ and $u_6$ receive the list $l_2=\{u_1, u_2, u_3, u_4, u_5, u_6\}$ with the signatures $\sigma_2=\{\sigma_1\{l_2\},$ $\sigma_2\{l_2\},$ $\sigma_5\{l_2\}, \sigma_6\{l_2\}\}$. Since the signatures are correctly made and $|\sigma_1| \geq t$, $|\sigma_2| \geq t$, no one will abort during the consistency check, and thus the input of client $u_6$ will be leaked. We can observe that client $u_1$ and $u_2$ have made their signatures on both $l_1$ and $l_2$. Therefore, even with $t \geq \lfloor  \frac{n}{2} \rfloor +1$, the privacy of clients' inputs can be compromised in the setting where both the server and clients are active malicious.

Now we proceed to discuss the lower bound of threshold $t$ in this model. Recall that we are considering the setting where they are $n$ clients in total including $n_c$ malicious clients with the threshold $t$. For two fraudulent lists $l_1$ and $l_2$, let the number of signatures on $l_1$ and $l_2$ from honest clients be $n_1$ and $n_2$ respectively, and the number of signatures from malicious clients be $n_c$. In order to have two fraudulent lists $l_1$ and $l_2$ passing the consistency check, the number of effective signatures on $l_1$ and $l_2$ should be not less than $t$ such that
$$t \leq n_1+n_c, t \leq n_2+n_c$$
while $n_1+n_2+n_c \leq n$,
we can observe that for $t \leq \frac{n+n_c}{2}$, it is possible to construct two different views to clients. Therefore, to guarantee the security, the threshold $t$ must be greater than $\frac{n+n_c}{2}$. Furthermore, as the number of adversaries $n_c$ cannot be greater than $\frac{n}{3}$ (refer to \cite{bonawitz2017practical} for the details), we can have the minimum threshold $t$ for the required security is $\lfloor  \frac{2n}{3} \rfloor +1$.

\section{Performance Analysis}
\label{sec:performance-analysis}

\begin{figure*}[htb]
\centering
\begin{minipage}[b]{0.48\linewidth}
\includegraphics[width=\linewidth]{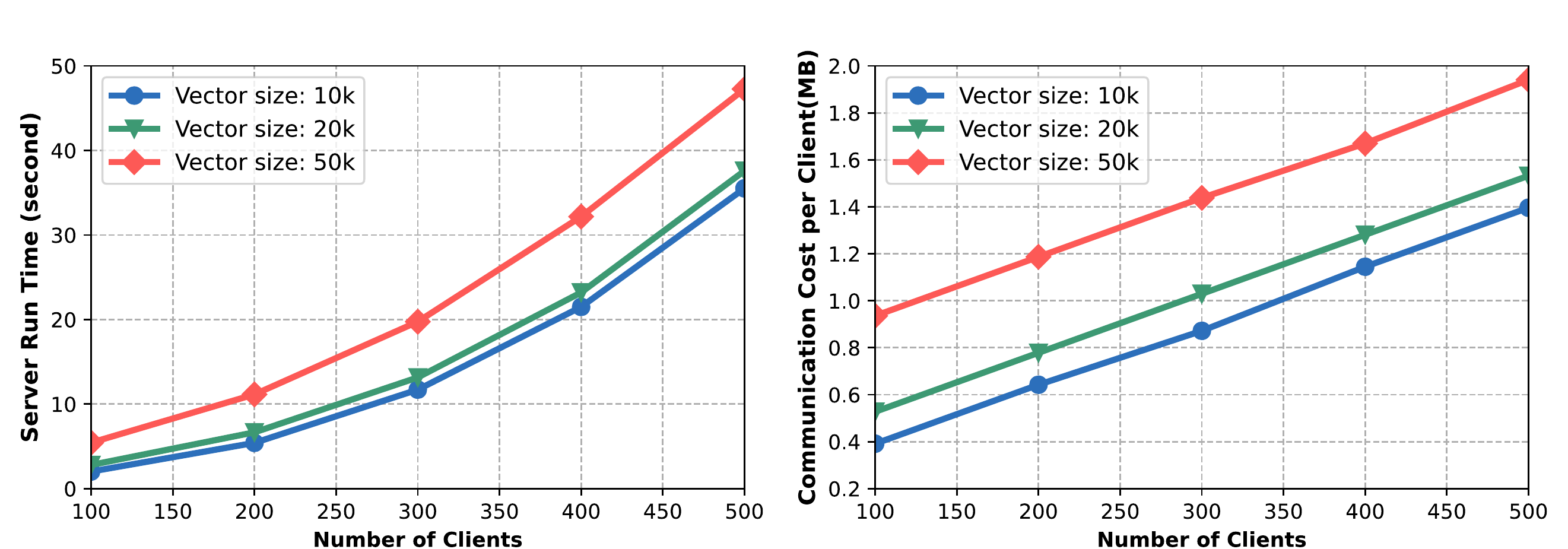}
\centering
\caption{Server's (total) runtime and client's communication cost \\as the number of clients increases. No dropout client.}
\label{fig:num_increase}
\end{minipage}
\hspace{0.2cm}
\begin{minipage}[b]{0.48\linewidth}
\centering
\includegraphics[width=\linewidth]{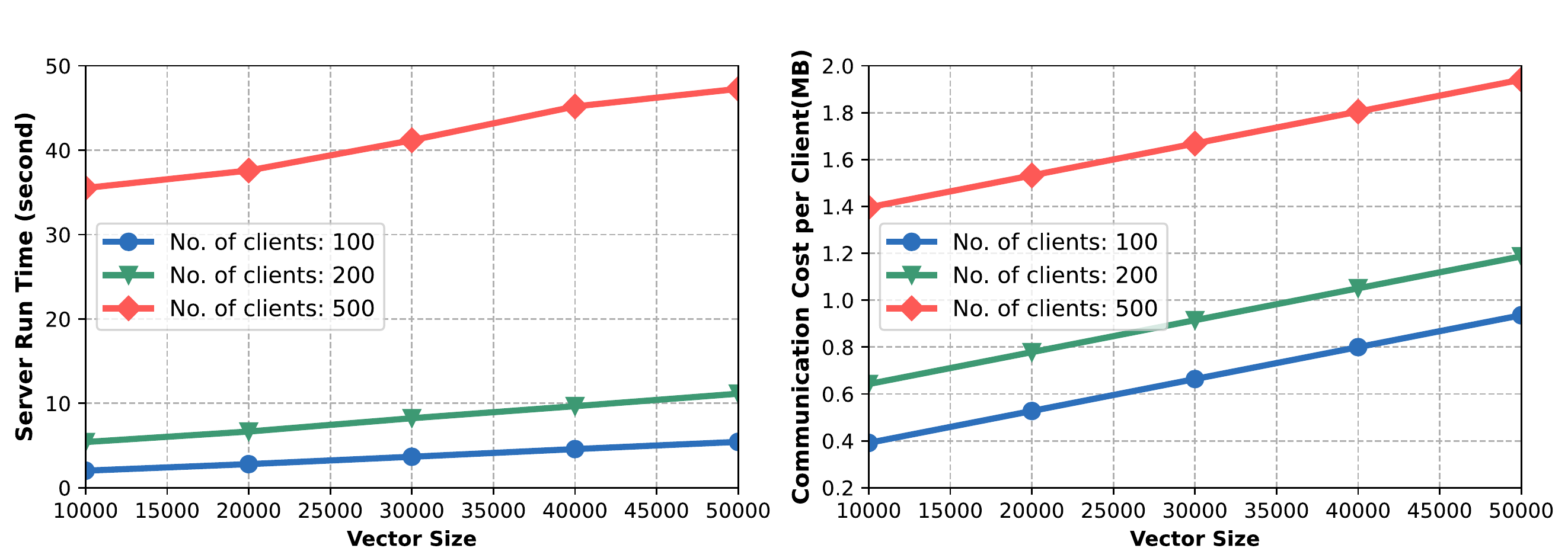}
\centering
\caption{Server's (total) runtime and client's communication cost as the vector size increases. No dropout client.}
\label{fig:vector_increase}
\end{minipage}
\end{figure*}

\begin{figure*}
\centering
\subfloat[Runtime in LAN setting.]{
\begin{minipage}[t]{1\linewidth}
\includegraphics[width=1\linewidth]{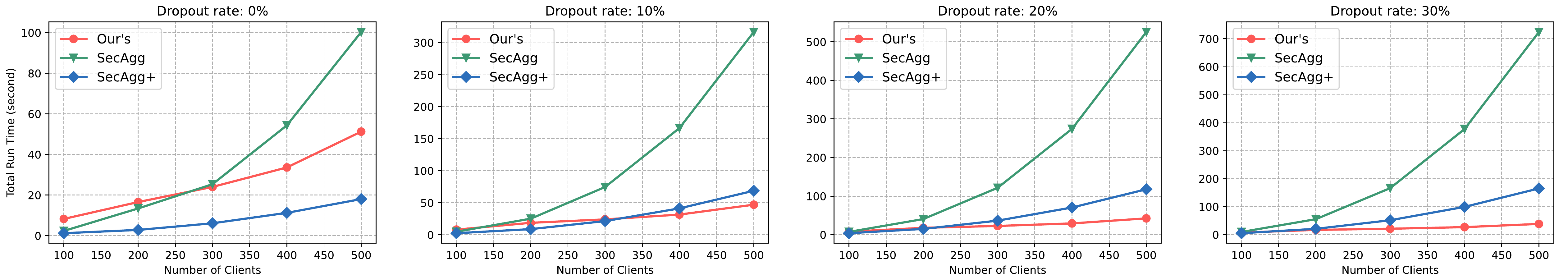}
\label{fig:dropout_runtime}
\end{minipage}%
}

\subfloat[Communication cost in LAN setting.]{
\begin{minipage}[t]{1\linewidth}
\includegraphics[width=1\linewidth]{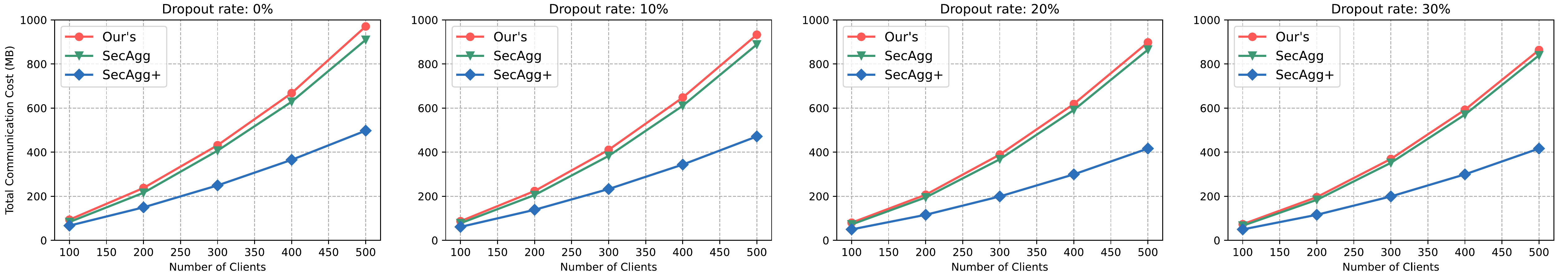}
\label{fig:dropout_communication_cost}
\end{minipage}%
}%

\subfloat[Runtime in WAN setting.]{
\begin{minipage}[t]{1\linewidth}
\includegraphics[width=1.0\linewidth]{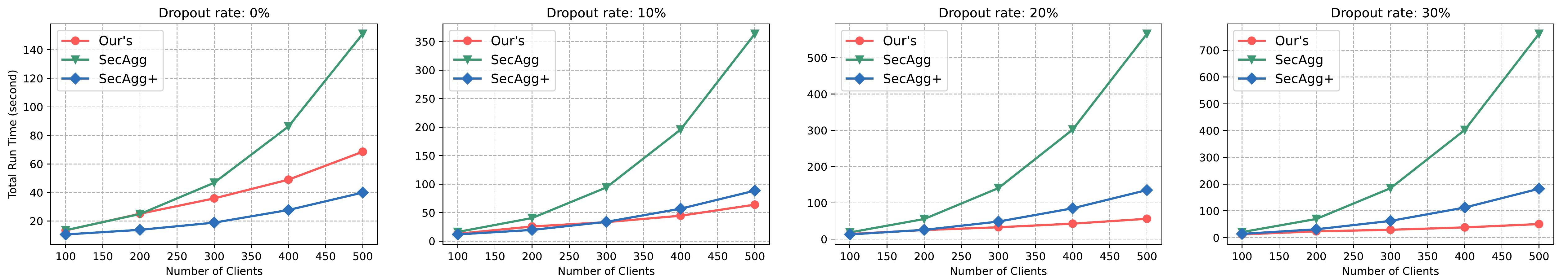}
\label{fig:dropout_runtime_wan}
\end{minipage}%
}

\subfloat[Communication cost in WAN setting.]{
\begin{minipage}[t]{1\linewidth}
\includegraphics[width=1.0\linewidth]{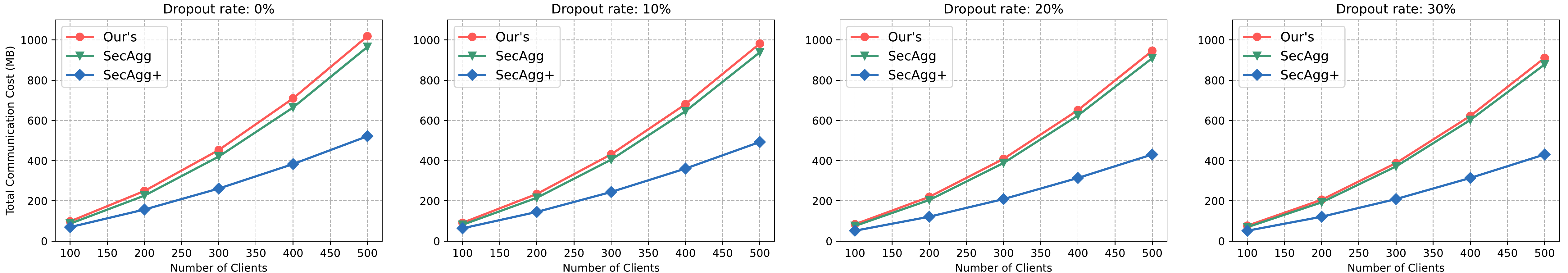}
\label{fig:dropout_communication_cost_wan}
\end{minipage}%
}%
\caption{Total runtime and communication cost with different dropout rates. The vector size is fixed to 50K.}
\label{fig:lan_wan_setting}
\end{figure*}

\begin{figure*}
\centering
\centering
\subfloat[LAN setting.]{
\begin{minipage}[t]{1\linewidth}
\includegraphics[width=1.0\linewidth]{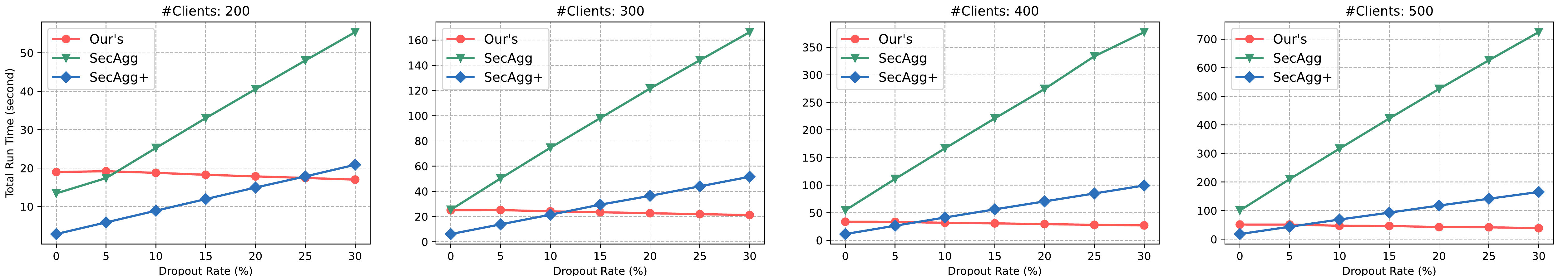}
\label{fig:drop-rate-subgraphs_lan}
\end{minipage}%
}

\subfloat[WAN setting.]{
\begin{minipage}[t]{1\linewidth}
\includegraphics[width=1.0\linewidth]{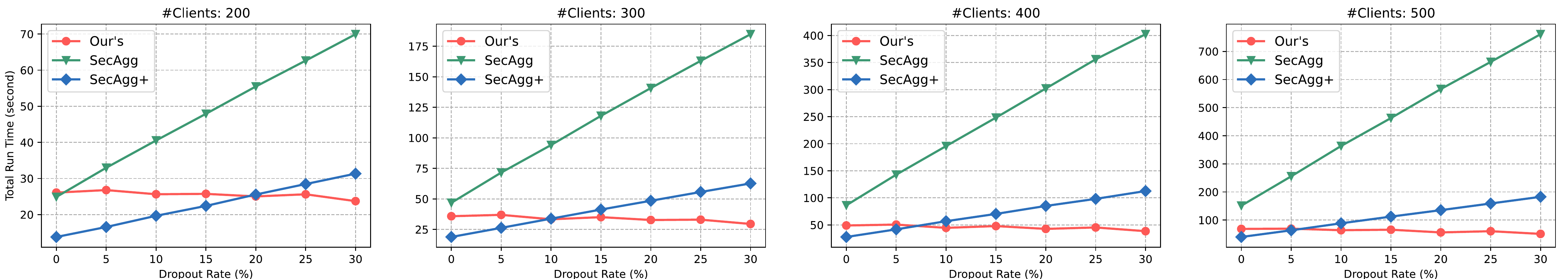}
\label{fig:drop-rate-subgraphs_wan}
\end{minipage}%
}

\caption{Total runtime as the dropout rate increases in different network settings. The vector size is fixed to 50K.}
\label{fig:drop-rate-subgraphs}
\end{figure*}

In this section, we first analyze the computation and communication cost of the client and server, respectively. Then, we implement a prototype to show the performance of our proposed protocol.

\subsection{Complexity analysis}

We summarize the analysis of complexity and dropout-resilience compared to other existing protocols in Table \ref{tab:comp_table}.

\textbf{Computation overheads:} Each client's computation cost mainly depends on (i) computing $n$ $(t, n)$ Shamir secret sharing with $\mathcal{O}(n^2)$ complexity and (ii) generating $m$ mask for the input vector using HPRG with $\mathcal{O}(m)$ complexity. Thus the total computation complexity per client is $\mathcal{O}(n^2+m)$. Since the server involves only one reconstruction from $n$ $(t, n)$ Shamir secret shares, the overhead is $\mathcal{O}(n)$. Note that we adopt similar method in SecAgg to precompute Lagrange basis polynomials (see Section \ref{sec:Shamir-secret-sharing} and \cite{bonawitz2017practical} for the details) for Shamir secret sharing, computational cost for one reconstruction results in $\mathcal{O}(n)$ complexity rather than $\mathcal{O}(n^2)$ in standard Shamir scheme.

\textbf{Communication overheads:} Communication cost of each client consists of sending $n$ encrypted shares and one masked model with $m$ elements, which causes $\mathcal{O}(m+n)$ complexity. The communication cost of the server is dominated by forwarding messages between every pair of clients with $\mathcal{O}(n^2)$ complexity and receiving masked models from the clients with $\mathcal{O}(mn)$ complexity, which is $\mathcal{O}(n^2+mn)$ complexity in total. Note that compared with SecAgg, our scheme does not involve pair-wise DH protocol, thus reducing $2n$ key exchange overhead for each client and $2n^2$ message forwarding for the server.

\textbf{Dropout resilience:} First, we note that the maximum number of dropout clients that our proposed scheme can tolerate is the same as that of the SecAgg scheme since both of them use a $(t,n)$ Shamir secret sharing scheme, which does not sacrifice any dropout-resilience compared to SecAgg+ scheme. Furthermore, by replacing communication-intensive Diffie-Hellman protocol adopted in SecAgg and SecAgg+ with HPRG and additive operation based on Shamir secret sharing scheme, the runtime of our proposed protocol decreases with the increase of the dropout rate, rather than the increasing runtime of previous works, which implies the stronger dropout-resilience of our scheme.

\subsection{Experiments}
\label{sec:experiments}
\begin{table*}[htb]
\centering
\caption{The client's and server's (total) runtime for different steps of the proposed protocol in LAN/WAN settings under semi-honest model. The vector size is fixed to 50K with 64 bits length.}
\begin{tabular}{ccccccc} 

\toprule[1pt]
& Num. clients & Dropout rate & Sharing seeds & Collecting Masked models & Unmasking & total runtime\\ 
\midrule
Client  & 500 & 0\% & 2.28s/7.76s & 0.78s/3.24s & 0.03s/0.84s & 3.09.s/11.84s      \\ 
Server  & 500 & 0\% & 2.61s/7.76s & 0.83s/3.24s & 44.45s/57.56s & 47.88s/68.56s     \\ 
Server  & 500 & 10\% & 2.47s/8.26s & 0.78s/3.25s & 38.96s/52.48s & 42.21s/63.99s     \\ 
Server  & 500 & 20\%  & 2.47s/6.63s & 0.78s/2.62s & 35.33s/46.88s & 38.59s/56.13s  \\     
Server  & 500 & 30\%  & 2.47s/6.16s & 0.79s/2.63s & 31.08s/42.11s & 34.34s/50.90s  \\ 
\midrule
Client  & 1000 & 0\% & 8.36s/16.77s & 0.78s/3.91s & 0.23s/1.85s & 9.37s/22.53s       \\ 
Server  & 1000 & 0\% & 8.81s/16.77s & 0.79s/3.91s & 145.34s/242.37s &  154.94s/263.05s      \\ 
Server  & 1000 & 10\% & 8.83s/14.15s & 0.78s/3.69s & 129.10s/224.01s &  138.71s/241.85s      \\ 
Server  & 1000 & 20\% & 8.87s/14.23s & 0.78s/3.57s & 108.17s/177.22s &  117.82s/195.02s      \\ 
Server  & 1000 & 30\% & 8.85s/17.31s & 0.83s/3.68s & 91.95s/147.71s &  101.63s/168.70s      \\ 

\bottomrule[1pt]
\end{tabular}
\label{tab:total_table}
\end{table*}

\begin{table*}[htb]
\centering
\caption{The client's and server's (total) runtime for different steps of the proposed protocol in LAN/WAN settings under active malicious model. The vector size is fixed to 50K with 64 bits length.}
\begin{tabular}{cccccccc} 

\toprule
       & \begin{tabular}[c]{@{}c@{}}Num. \\clients\end{tabular} & \begin{tabular}[c]{@{}c@{}}Dropout \\rate\end{tabular} & \begin{tabular}[c]{@{}c@{}}Sharing \\seeds\end{tabular} & \begin{tabular}[c]{@{}c@{}}Collecting Masked \\models\end{tabular} & \begin{tabular}[c]{@{}c@{}}Checking \\consistency\end{tabular} & Unmasking       & Total runtime     \\ 
\midrule
Client  & 500 & 0\% & 5.08s/9.54s & 0.84s/2.65s & 0.03s/0.87s & 0.23s/0.12s & 6.18s/13.18s       \\ 
Server  & 500 & 0\% & 5.26s/9.54s & 0.83s/2.65s & 0.03s/0.87s & 88.58s/131.14s & 94.70s/144.20s     \\ 
Server  & 500 & 10\% & 5.26s/11.55s & 0.85s/3.69s & 0.03s/0.87s & 73.23s/106.59s & 79.36s/122.70s     \\ 
Server  & 500 & 20\%  & 5.26s/11.46s & 0.84s/3.72s & 0.03s/0.87s & 60.34s/90.28s & 66.48s/106.33s  \\     
Server  & 500 & 30\%  & 5.26s/12.75s & 0.83s/4.34s & 0.02s/0.86s & 48.90s/75.07s & 55.03s/93.03s  \\ 
\midrule
Client  & 1000 & 0\% & 14.79s/20.19s & 0.79s/3.33s & 0.10s/0.95s & 5.75s/14.87s & 21.43s/39.34s       \\ 
Server  & 1000 & 0\% & 14.81s/20.19s & 0.79s/3.33s & 0.10s/0.95s & 406.40s/594.11s &  422.10s/619.30s      \\ 
Server  & 1000 & 10\% & 14.85s/22.53s & 0.79s/3.04s & 0.09s/0.94s & 311.71s/482.24s & 327.44s/508.75s      \\ 
Server  & 1000 & 20\% & 14.88s/23.79s & 0.79s/3.61s & 0.09s/0.93s & 238.46s/391.50s & 254.21s/419.39s      \\ 
Server  & 1000 & 30\% & 14.83s/23.65s & 0.79s/3.44s & 0.08s/0.92s & 174.12s/307.38s & 189.81s/335.39s      \\ 

\bottomrule[1pt]
\end{tabular}
\label{tab:total_table_malicious}
\end{table*}

Our prototype is tested on two c5.4xlarge AWS EC2 instances running Ubuntu 18.04 with 16 vCPUs and 32GB memory. The central server executes on one instance while the clients run parallel in the other instance. In the LAN setting, the instances are both hosted within the same region, i.e., Singapore (ap-southeast-1), with 3.72ms latency and 4.80Gbps bandwidth on average. In the WAN setting, the central server executes on an instance located in Singapore, and the clients execute on another instance located in northern Virginia, US (us-east-1), with 211.31ms latency and 4.18Gbps bandwidth on average. Our main scheme is implemented in Python using Gmpy2 \cite{horsen2016gmpy2}, Cryptography library \cite{cryptolib}, and several other standard libraries. The Pollard's lambda method for computing discrete logarithms is implemented in C++ using NTL library \cite{shoup2001ntl}. For cryptographic primitives, we adopt AES-GCM with 128-bit keys for authenticated encryption, standard $(t, n)$ Shamir secret sharing scheme to deal with dropped clients, and DDH based HPRG constructed with an SHA-256 hash to generate masks. We note that our proposed scheme is usable for all kinds of machine learning models, as we only focus on the part for secure aggregation. Thus, our scheme does not affect any performance of the ML model, and there is no deviation due to the usage of cryptographic techniques. Our implementation is available as open-source at \cite{ourcode}.

As noted earlier, our focus is on PPML systems that clients are resource-constrained mobile devices, hence it is necessary to investigate the communication cost from the client perspective. In addition, since the new FL round begins only after the server obtains the aggregation result in the last FL round, the total runtime from the server perspective may also affect user experience. Thus, we first evaluate the server's runtime and client's communication cost with different numbers of clients and vector sizes in a no-dropout setting, of which the experimental results are given in Fig. \ref{fig:num_increase} and Fig. \ref{fig:vector_increase}. We can observe that with 500 clients and 50K vector size comparable to LeNet \cite{lecun1998gradient}, only about 2MB communication cost is required for each client, and the whole aggregation can be done in about 1 minute, which implies the efficiency of our proposed scheme.

Furthermore, taking dropout clients into account, we investigate the total runtime and client's communication cost with different dropout rates in both LAN and WAN settings, of which the experimental results are given in Fig. \ref{fig:lan_wan_setting}. We can observe that our scheme involves a comparable total communication cost of the whole aggregation to that of SecAgg, as the same $(t,n)$ Shamir scheme is kept, but with less total runtime. Such observation keeps consistency in both LAN and WAN settings. Specifically, for large-scale systems, say with 500 clients, when the dropout rate is large, say greater than $10\%$, our scheme's efficiency performance outperforms SecAgg and SecAgg+. Fig. \ref{fig:drop-rate-subgraphs} also supports this point. The total runtime of our proposed scheme decreases with the increase of the dropout rate, rather than the increasing runtime of SecAgg and SecAgg+, which implies the stronger dropout-resilience of our scheme compared to previous works.

Besides, the client's and server's runtime for different steps of our proposed protocol under semi-honest setting and active malicious setting are given in Table. \ref{tab:total_table} and Table. \ref{tab:total_table_malicious} respectively, from which we can observe the additional cost to guarantee the security under the active malicious setting due to the use of a large number of signatures and verification techniques. Moreover, as mentioned in Section \ref{sec:generating-masks}, the server can accelerate the computation of discrete logarithms in the unmasking step of our protocol by using Pollard's lambda method \cite{menezes2018handbook}. In Table. \ref{tab:dlp}, we show the runtime of such computation for different sizes of the clients' gradients comparable to several ML models. We can observe from Table. \ref{tab:dlp} that the time required is competitive for small traditional ML models and simple neural networks such as LeNet \cite{lecun1998gradient}, but becomes impractical for large-scale neural networks such as ResNet18 \cite{he2016deep}. However, we note that for lightweight neural networks such as MobileNet V3 \cite{howard2017mobilenets} that can be deployed on resource-constrained mobile devices, the overheads of computing discrete logarithms can be still affordable if more powerful servers and multi-threading implementations are adopted.

\begin{table}
\centering
\caption{Runtime of computing discrete logarithms using Pollard's lambda method with different vector size. The order of involved finite cyclic group is set to be a 1024-bit prime.}
\label{tab:dlp}
\begin{tabular}{ccc} 
\toprule[1pt]
Vector Size & Model Type & Runtime    \\ \midrule
72 & Linear Regression\cite{candanedo2017data} &  30ms\\
7850 & Logistic Regression \cite{lecun1998gradient} &  1.9s\\
35K & SVM \cite{svm} & 8.6s \\
50K & LeNet \cite{lecun1998gradient} & 12.5s    \\
2.5M & MobileNet V3 small \cite{howard2017mobilenets} & 601.5s \\ 
5.5M & MobileNet V3 large \cite{howard2017mobilenets} & 1338.7s\\
11.7M & ResNet18 \cite{he2016deep} & 2039.3s\\
\bottomrule[1pt]
\end{tabular}
\end{table}

\subsection{Further discussions}

As noted in Section \ref{sec:experiments}, we emphasize that our focus is on designing a secure aggregation protocol to protect the privacy of clients' gradient vectors, hence the intermediate and final global models are revealed to all participants. Thus, our proposed scheme is still vulnerable to the membership inference attack \cite{shokri2017membership}. In this case, attackers can determine if a record is in clients' training datasets, given only some global models. Protecting the privacy of global models requires clients to train their ML models over encrypted global models. Existing solutions include HE-based schemes \cite{FroelicherTPSSB21,sinem21poseidon} and MPC-based schemes\cite{ChaudhariRS20,liu2020mpc}. However, those solutions may involve large overheads for large-scale ML models such as deep neural networks. To improve the efficiency while keeping the privacy of global models, sophisticated integrations of our proposed scheme with existing solutions are required.

Furthermore, our scheme can be integrated with the defense methods against so-called poisoning attack or backdoor attack \cite{bagdasaryan2020backdoor,nguyen2020poisoning} to further improve its security. These defense methods usually evaluate FL updates based on a well-designed metric to detect the poisoned updates. However, only methods with a proper selection of updates to be evaluated can be integrated with our scheme straightforwardly. For example, BaFFLe \cite{andreina2021baffle} avoids backdoor attacks by validating the new global model to be updated, which does not leak any information of clients' models to the server, and thus can be adopted for integration. In contrast, defense methods such as \cite{shen2016auror,blanchard2017machine,nguyen2021flguard,CaoF0G21} rely on the evaluation of clients' locally trained models, which means that the server must know clients' models. This directly leads to a breach of the security requirements of PPML that preserve the privacy of clients' models hence data. Therefore, those defense methods are hindered from deployments for the integration with our scheme and other secure aggregation schemes.

\section{Conclusions}
\label{sec:conclusion}

We proposed an efficient aggregation protocol to compute the sum of inputs from a set of participants while preserving their input privacy. Our protocol allows participants to drop out from the protocol during the execution and provides stronger dropout-resilience compared to previous works. Thus it is suitable to be applied to large-scale PPML scenarios. Additionally, the security of our protocol is guaranteed against both semi-honest and active malicious adversaries by setting proper system parameters. Besides, the simplicity of the proposed scheme makes it attractive both for implementation and for further improvements.

\section*{Acknowledgments}

We would like to thank three anonymous reviewers for their comments on the earlier versions of this paper. In addition, we thank the authors of SAFELearn \cite{fereidooni2021safelearn}, TurboAgg \cite{so2021turbo} and FastSecAgg \cite{kadhe2020fastsecagg} for their shared materials and helpful discussions.

\bibliographystyle{IEEEtran}
\bibliography{IEEEfull,template}
%

\newpage

\begin{IEEEbiography} [{\includegraphics[width=1in,height=1.25in,clip,keepaspectratio]{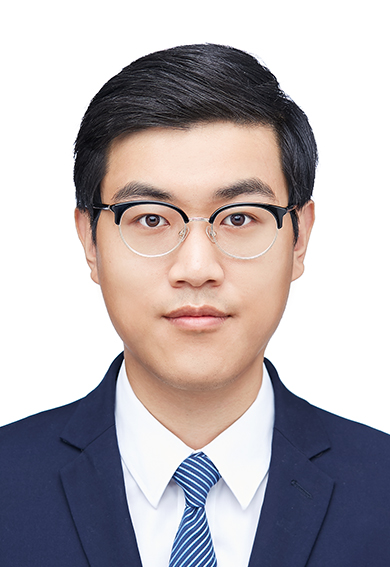}}]{Ziyao Liu}
received his B.E. degree from the school of Electronics Information Engineering, Zhengzhou University, Zhengzhou, China, in 2015, and the M.S. degree from Beijing Institute of Technology, Beijing, China, in 2018. He is currently working towards a Ph.D. degree in the School of Computer Science and Engineering, Nanyang Technological University, Singapore. His research interests include privacy-preserving machine learning, multi-party computation, and applied cryptography.
\end{IEEEbiography}

\begin{IEEEbiography} [{\includegraphics[width=1in,height=1.25in,clip,keepaspectratio]{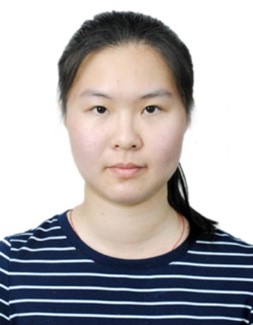}}]{Jiale Guo}
received her B.S. from the School of Mathematics, Shandon University, Jinan, China, in 2017. She is currently pursuing a Ph.D. degree in the School of Computer Science and Engineering, Nanyang Technological University, Singapore. Her research interests include privacy-preserving machine learning and Cybersecurity. 
\end{IEEEbiography}

\begin{IEEEbiography} [{\includegraphics[width=1in,height=1.25in,clip,keepaspectratio]{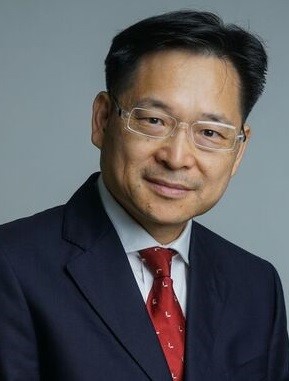}}]{Kwok-Yan Lam}
(Senior Member, IEEE) received his B.Sc. degree (1st Class Hons.) from University of London, in 1987, and Ph.D. degree from University of Cambridge, in 1990. He was a Visiting Scientist at the Isaac Newton Institute, Cambridge University, and a Visiting Professor at the European Institute for Systems Security. He has collaborated extensively with law-enforcement agencies, government regulators, telecommunication operators, and financial institutions in various aspects of Infocomm and Cyber Security in the region. From 2002 to 2010, he was a Professor with Tsinghua University, China. Since 1990, he has been a Faculty Member with the National University of Singapore and the University of London. He is currently a Full Professor with Nanyang Technological University, Singapore and the Director of the Strategic Centre for Research in Privacy-Preserving Technologies and Systems (SCRiPTS). From August 2020, Professor Lam is also on part-time secondment to the INTERPOL as a Consultant at Cyber and New Technology Innovation. In 1998, he received the Singapore Foundation Award from the Japanese Chamber of Commerce and Industry in recognition of his research and development achievement in information security in Singapore.
\end{IEEEbiography}

\begin{IEEEbiography} [{\includegraphics[width=1in,height=1.25in,clip,keepaspectratio]{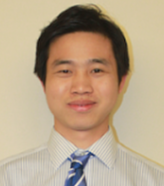}}]{Jun Zhao}
(S'10-M'15) is currently an Assistant Professor in the School of Computer Science and Engineering (SCSE) at Nanyang Technological University (NTU), Singapore. He received a Ph.D. degree in Electrical and Computer Engineering from Carnegie Mellon University (CMU), Pittsburgh, PA, USA, in May 2015, and a bachelor's degree in Information Engineering from Shanghai Jiao Tong University, China, in June 2010. One of his papers was a finalist for the best student paper award in IEEE International Symposium on Information Theory (ISIT) 2014. His research interests include A.I. and data science, security and privacy, control and learning in communications and networks.
\end{IEEEbiography}

\vfill

\end{document}